\documentclass[10pt,a4paper,dvipsnames,usenames]{article}
\usepackage{amssymb}
\usepackage{amsmath}
\usepackage{amsthm}
\usepackage{latexsym}
\usepackage[dvips]{epsfig}
\usepackage{mathrsfs}
\usepackage{eufrak}
\usepackage{bm}
\usepackage{stmaryrd}
\usepackage{authblk}
\usepackage{slashed}

\usepackage[dvipsnames]{xcolor}
\usepackage{tikz}
\theoremstyle{plain}
\newtheorem{proposition}{Proposition}

\newtheorem{theorem}{Theorem}

\newtheorem{corollary}{Corollary}
\newtheorem*{main}{Theorem}

\newtheorem{remark}{Remark}

\def\bma{{\bm a}}
\def\bmb{{\bm b}}
\def\bmc{{\bm c}}
\def\bmd{{\bm d}}
\def\bme{{\bm e}}
\def\bmf{{\bm f}}
\def\bmg{{\bm g}}
\def\bmh{{\bm h}}
\def\bmi{{\bm i}}
\def\bmj{{\bm j}}
\def\bmk{{\bm k}}
\def\bml{{\bm l}}
\def\bmn{{\bm n}}

\def\bmu{{\bm u}}

\def\bmx{{\bm x}}

\def\bmzero{{\bm 0}}
\def\bmone{{\bm 1}}
\def\bmtwo{{\bm 2}}
\def\bmthree{{\bm 3}}

\def\bmK{{\bm K}}
\def\bmL{{\bm L}}

\def\bmS{{\bm S}}
\def\bmT{{\bm T}}
\def\bmX{{\bm X}}

\def\bmalpha{{\bm \alpha}}
\def\bmbeta{{\bm \beta}}
\def\bmgamma{{\bm \gamma}}

\def\bmomega{{\bm \omega}}

\def\bmphi{{\bm \phi}}

\def\bmtau{{\bm \tau}}
\def\bmupsilon{{\bm \upsilon}}

\def\bmGamma{{\bm \Gamma}}

\def\bmpartial{{\bm \partial}}
\def\bmnabla{{\bm \nabla}}

\usepackage{mathtools}% http://ctan.org/pkg/mathtools
\usepackage{xparse}% http://ctan.org/pkg/xparse
\makeatletter
\newcommand{\raisemath}[1]{\mathpalette{\raisem@th{#1}}}
\newcommand{\raisem@th}[3]{\raisebox{#1}{$#2#3$}}
\makeatother
\NewDocumentCommand{\newrbar}{O{0pt} O{0pt}}{
  \ensuremath{\mathrlap{\raisemath{#2}{\hspace*{#1}{\mathchar'26\mkern-9mu}}}r}}

\newcounter{mnotecount}%[section]

\newcommand{\mnotex}[1]%{}
{\protect{\stepcounter{mnotecount}}$^{\mbox{\footnotesize $\bullet$\themnotecount}}$ 
\marginpar{%\color{red}%
\raggedright\tiny\em
$\!\!\!\!\!\!\,\bullet$\themnotecount: #1} }

\newcounter{mnote}

\usepackage{color}
\usepackage{xcolor}
\usepackage[normalem]{ulem}
\usepackage{soul}
\usepackage{hyperref}

\begin{document}

\title{\textbf{A conformal approach to the stability of Einstein
    spaces with spatial sections of negative scalar curvature}}
 
\author[1]{ Marica Minucci  \footnote{E-mail
    address:{\tt m.minucci@qmul.ac.uk}}}
\author[1]{Juan A. Valiente Kroon \footnote{E-mail address:{\tt j.a.valiente-kroon@qmul.ac.uk}}}
%\author[1]{Con T. Ributor}

\affil[1]{School of Mathematical Sciences, Queen Mary, University of London,
Mile End Road, London E1 4NS, United Kingdom.}

\maketitle

\begin{abstract}
In this article, it is shown how the extended conformal Einstein field
equations and a gauge based on the properties of conformal geodesics
can be used to analyse the non-linear stability of de Sitter-like spacetimes
with spatial sections of negative scalar curvature.  This class of
spacetimes admits a smooth conformal extension with a space-like conformal
boundary. Central to the analysis is the use of conformal Gaussian
systems to obtain a hyperbolic reduction of the conformal Einstein
field equations for which standard Cauchy stability results for
symmetric hyperbolic systems can be employed. The use of conformal methods
allows us to rephrase the question of the global existence of solutions to
the Einstein field equations into considerations of finite existence
time for the conformal evolution system.
\end{abstract}

%\tableofcontents

\section{Introduction}
In the Mathematical Relativity literature, for a \emph{Cosmological
  spacetime} it is usually understood a spacetime with compact spatial
sections. Understanding the long-time evolution of generic examples of
these spacetimes in, say the \emph{vacuum case}, is one of the open challenges in the area. Although
generic initial data is expected to form singularities towards the future, it is
nevertheless essential to address the stability of those solutions
which are known to be geodesically complete. The fundamental example
of a geodesically complete Cosmological spacetime is given by the
\emph{de Sitter spacetime}. Its non-linear stability was analysed in the
seminal work by Friedrich \cite{Fri86b,Fri86c}. A central aspect of
this result is the use of conformal methods to transform the question of
the global existence of solutions to a finite existence problem. An
alternative approach to the study of the non-linear stability of
vacuum Cosmological solutions to the Einstein field equations by means
of so-called \emph{CMC foliations} has been used by Andersson 
and Moncrief \cite{AndMon03,AndMon04}  to prove the non-linear stability of 4-dimensional
Friedmann-Lema\^{i}tre-Robinson-Walker (FLRW) vacuum solutions. Using
similar methods, in \cite{FajKro20} Fajman and Kr\"oncke studied the non-linear
stability of large classes of Cosmological solutions to the vacuum Einstein field
equations with a positive Cosmological constant in arbitrary
dimensions. These solutions are characterised by having spatial
sections with constant scalar curvature which can be either positive
or negative. \emph{The purpose of this article is to show that, in four
dimensions, the stability results for spacetimes with spatial sections
of constant negative curvature given in \cite{FajKro20}
can be addressed via a generalisation of the conformal methods developed by Friedrich
\cite{Fri86b,Fri91,Fri95,Fri15} ---see also \cite{CFEBook}.} The
analysis of the case positive constant curvature is essentially
contained in the original results in \cite{Fri86b} ---see also
\cite{LueVal09}. The use of conformal methods in the stability problem
considered in this article provides alternative information and
insights into the evolution of Cosmological spacetimes. 

\subsubsection*{Conformal methods in the analysis of de Sitter-like spacetimes}

In what follows, for a \emph{de Sitter-like spacetime} it is
understood a vacuum Einstein spacetime with a positive value of the
Cosmological constant and compact spatial sections.  General results
on conformal geometry show that \emph{if} these spacetimes admit a conformal
compactification \emph{\`{a} la} Penrose then the conformal boundary
of the spacetime must be space-like ---see e.g. \cite{CFEBook},
Theorem 10.1. Following the standard usage, we refer to the conformal
extension of a de Sitter-like spacetime as the \emph{unphysical
  spacetime}. The usefulness of this conformal extension lies in the
fact that points representing the infinity of the physical spacetime
(e.g. the endpoints if time-like geodesics) are mapped to a finite
location in the unphysical spacetime. These points are characterised by the vanishing of the conformal
factor. 

In the particular case considered in the present article, we
consider de Sitter-like spacetimes which can be conformally embedded
into a portion of a cylinder whose sections have negative scalar
curvature. The conformal embedding is realised by means of a conformal
factor $\Theta$ which depends quadratically on the affine parameter
$\tau$ of special curves which are invariants of the conformal
structure. These curves are known as \emph{conformal geodesics}, and
the affine parameter is used as a time coordinate for the physical
metric. 

Key in the conformal approach is that the unphysical metric provides a
solution to the \emph{conformal Einstein field equations} ---i.e. a
conformal representation of the vacuum Einstein field equations which
provides equations which are regular up and beyond the conformal
boundary \cite{Fri84,CFEBook}. These equations allow to avoid the difficulties produced by
the fact that the direct application of conformal transformation laws
into the Einstein equations leads to equations with singular terms. In
the present article we make use of a more general version of these
equations, the \emph{extended conformal Einstein field equations}
expressed in terms of a \emph{Weyl connection} ---i.e. a non-metric
torsion free connection which preserves the causal structure. This
version of the conformal equations allows the use of \emph{conformal Gaussian coordinate
systems} in which coordinates are propagated along conformal geodesics
---rather than along standard geodesics as is done in the usual
Gaussian systems. 

As already mentioned, the appeal of conformal methods in the study of
solutions to the Einstein field equations lies in the observation that
local results for the unphysical spacetime can, in principle, be
translated into global results for the physical
spacetime. In original formulation of the conformal Einstein field
equations the conformal factor realising the conformal embedding of
the physical spacetime into a compact manifold is an unknown of the
problem. However, remarkably, the use of conformal Gaussian coordinate
systems provide a natural conformal factor which
singles out a representative in the conformal class of the
spacetime. Accordingly, the location of the conformal boundary is
known \emph{a priori}, thus simplifying further the analysis of the
evolution equations. The extended conformal Einstein field equations expressed
in terms of a conformal Gaussian system can be shown to imply a
conformal evolution system which takes the form of a symmetric
hyperbolic system ---i.e. a class of evolution systems for which there
exists a well-developed existence, uniqueness and stability theory
\cite{Kat75}. 

\subsubsection*{The main result}

In the following, let $(\mathcal{S},\mathring{\bmgamma})$ denote a
compact and complete 3-dimensional Riemannian
manifold with negative constant curvature. Then the Lorentzian metric
given 
\begin{equation}
\label{BackgroundMetricIntro}
\mathring{\tilde{\bmg}} = -\mathbf{d}t\otimes\mathbf{d}t + \sinh^2t \mathring{\bmgamma}
\end{equation}
is an Einstein space over $\mathbb{R}\times\mathcal{S}$ which is
geodesically complete and for which the Cosmological constant takes
the value $\lambda=3$.  Our main result can be formulated, informally, as follows:

\begin{main}
Given smooth initial data $(\bmh,\bmK)$ for the Einstein field
equations on $\mathcal{S}$ which is 
suitably close (as measured by a suitable Sobolev norm) to the data
implied by the metric \eqref{BackgroundMetricIntro}, there exists a smooth metric
$\tilde{\bmg}$ defined over $[0,\infty)\times\mathcal{S}$ which is
close to $\mathring{\tilde{\bmg}}$ (again, in the sense of Sobolev
norms) and solves the vacuum Einstein field equations with
Cosmological constant $\lambda=3$\footnote{The choice of $\lambda=3$
  has been made to ease the presentation as it is customary in the
  mathematical Relativity literature. Howewer, our results readily
  extend to any (positive) value of the Cosmological constant. Indeed,
by means of a suitable constant conformal rescaling of the metric one
can obtain a solution to the vacuum Einstein field equations for an 
arbitrary positive value of the Cosmological constant ---see Remark
\ref{Remark:ValueLambda} for more details. }. The spacetime $([0,\infty)\times\mathcal{S},
\mathring{\tilde{\bmg}})$ is future geodesically complete.
\end{main}

A precise formulation of the result is given in Theorem \ref{Theorem:Main} in
Section \ref{Section:MainTheorem}. The construction of the initial
data required in the above result has been analysed in
\cite{ValWil20}. 

\medskip
This article is part of the general programme in mathematical
Relativity to understand the endpoint of the evolution of
“Cosmological spacetimes” (i.e. spacetimes with compact sections)
under the Einstein field equations (the so-called Einstein flow) and
it should be understood in this context. In particular, it identifies
a class of of spacetimes for which it is possible to show non-linear
stability and the existence of a regular conformal
representation. These special properties are not shared by generic
Cosmological solutions. Thus, it is important to identify the
situations for which this is the case. Moreover, our analysis may
allow to identify classes of solutions for which it may possible to
show global existence without the requirement of smallness used in our
perturbative analysis. Given that current Cosmological observations
indicate that the spatial sections of our Universe are flat, we do not
expect any direct physical application to physical Cosmology. 

\subsection*{Outline of the article}
The present article is structured as follows: Section
\ref{Section:XCFE} provides the required background on the extended
conformal Einstein field equations required for the analysis in this
article ---this discussion is not only restricted to the equations but
also involves the associated constraint equations and the notion of
conformal geodesics which will be used to fix the gauge. Section
\ref{Section:BackgroundSolution} provides an analysis of the
background spacetimes (Einstein spaces with spatial sections of
constant negative curvature) in the light of the conformal Einstein
field equations. In particular, this section gives a conformal
extension of these spacetimes arising from a certain class congruences of
conformal geodesics. Section \ref{Section:EvolutionEqns} provides an
analysis of the conformal evolution system which will be used in the
main stability argument and its relation to the actual extended
conformal Einstein field equations, including the so-called
\emph{propagation of the constraints}. Section
\ref{Section:InitialData} contains a brief discussion of the initial
data for the conformal evolution equations and how it can be
constructed. Section \ref{Section:Existence} contains the main
existence and stability analysis of the conformal evolution
equations. Section \ref{Section:GeodesicCompleteness} provides a discussion of the future geodesic
completeness of the perturbed spacetimes. Finally, Section \ref{Section:MainTheorem}
contains a precise statement of the main Theorem of this article and
some concluding remarks. In addition, the article contains two
appendices: Appendix \ref{Appendix:KatoThm} provides a summary of the
main technical tool in this article ---Kato's existence and stability
result for symmetric hyperbolic systems. Appendix
\ref{Appendix:GeodesicCompleteness} provides a brief discussion of the
geodesic completeness of the background solutions.  

\subsection*{Notations and conventions}
Throughout we mostly follow the notations and conventions of \cite{CFEBook}
except from the fact that the sign of the Cosmological constant for de
Sitter-like spacetimes is taken to be positive. The signature of
Lorentzian metrics is taken to be $(-+++)$.  Throughout the Latin
letters $a,\,b,\,c,\,\ldots$ denote spacetime abstract indices
indicating the tensorial character of the various objects while the
letters $i,\,j,\,k,\ldots$ correspond to spatial abstract indices. The
boldface indices $\bma,\,\bmb,\,\bmc,\ldots$ will be used as spacetime
frame indices ranging $\bmzero,\, \bmone,\,\bmtwo, \,\bmthree$ while
$\bmi,\, \bmj,\, \bmk,\ldots$ range over $\bmone,\, \bmtwo,
\,\bmthree$. The Greek indices $\mu,\, \nu,\, \lambda,\ldots$ play the
role of spacetime coordinate indices and $\alpha, \, \beta,\,
\gamma,\ldots$ are spatial coordinate indices. In addition to the
index notation described above, when convenient, we also make use of
an index-free notation ---e.g. a metric tensor can be described,
alternatively, by $\bmg$ or $g_{ab}$. Associated to a given metric
$\bmg$ we also make use of the \emph{musical isomorphisms} ${}^\sharp$
and ${}^\flat$ to denote the raising and lowering of indices of
tensorial objects away from their natural position.

Our conventions for the
curvature are given by the equation
\[
\nabla_a\nabla_b v^c -\nabla_b\nabla_a v^c = R^c{}_{dab} v^d.
\]

\section{The extended conformal Einstein field equations}
\label{Section:XCFE}
The main technical tool of this article is given by the \emph{extended
  conformal Einstein field equations} ---see \cite{Fri95,Fri98a}; also
\cite{CFEBook}. This system of equations constitute a conformal
representation of the vacuum Einstein field equations written in terms
of \emph{Weyl connections}. A solution to the extended conformal
equations implies a solution to the vacuum Einstein field equations
away from the conformal boundary.  In this section we provide a brief discussion of 
this system geared towards the applications of this article. A derivation and further discussion of the
general properties of these equations can be found in \cite{CFEBook},
Chapter 8.

\medskip
Throughout this article let $(\tilde{\mathcal{M}},\tilde{\bmg})$ with
$\tilde{\mathcal{M}}$  a 4-dimensional manifold and $\tilde{\bmg}$ a
Lorentzian metric denote a vacuum spacetime satisfying the Einstein field equations with
Cosmological constant
\begin{equation}
\tilde{R}_{ab} =\lambda \tilde{g}_{ab}.
\label{EFE}
\end{equation}
Let $\bmg$ denote an unphysical Lorentzian metric conformally related
to $\tilde\bmg$ via the relation
\[
\bmg = \Xi^2 \tilde\bmg
\]
with $\Xi$ a suitable conformal factor. Let $\nabla_a$ and
$\tilde{\nabla}_a$ denote, respectively, the Levi-Civita connections
of the metrics $\bmg$ and $\tilde\bmg$. 

\subsection{Weyl connections}
A Weyl connection is a torsion-free connection $\hat{\nabla}_a$ such
that
\[
\hat{\nabla}_a g_{bc} =-2 f_a g_{bc}.
\]
It follows from the above that the connections $\nabla_a$ and
$\hat{\nabla}_a$ are related to each other by
\begin{equation}
\hat{\nabla}_av^b-\nabla_av^b = S_{ac}{}^{bd}f_dv^c, \qquad
S_{ac}{}^{bd}\equiv \delta_a{}^b\delta_c{}^d +
\delta_a{}^d\delta_c{}^b-g_{ac}g^{bd},
\label{WeylToUnphysical}
\end{equation}
where $f_a$ is a fixed smooth covector and $v^a$ is an arbitrary
vector.  Given that 
\[
\nabla_a v^b -\tilde{\nabla}_a v^b =  S_{ac}{}^{bd} (\Xi^{-1} \nabla_a\Xi)v^c,
\]
one has that 
\[
\hat{\nabla}_av^b -\tilde{\nabla}_av^b =  S_{ac}{}^{bd}\beta_d v^c,
\qquad \beta_d \equiv f_d +\Xi^{-1}\nabla_d\Xi.
\]
In the following, it will be convenient to define
\begin{equation}
d_a \equiv \Xi f_a + \nabla_a \Xi. 
\label{Definition:CovectorD}
\end{equation}

\medskip
In the following $\hat{R}^a{}_{bcd}$ and $\hat{L}_{ab}$ will denote,
respectively, the Riemann tensor and Schouten tensor of the Weyl
connection $\hat{\nabla}_a$. Observe that for a generic Weyl
connection one has that $\hat{L}_{ab}\neq \hat{L}_{ba}$. One has the
decomposition
\[
\hat{R}{}^c{}_{dab} = 2 S_{d[a}{}^{ce}\hat{L}_{b]e} + C^c{}_{dab},
\] 
where $C^c{}_{dab}$ denotes the conformally invariant \emph{Weyl
  tensor}. The (vanishing) torsion of $\hat{\nabla}_a$ is denoted by
$\hat{\Sigma}_\bma{}^\bmc{}_\bmb$. In the context of the conformal
Einstein field equations it is convenient to define the \emph{rescaled
  Weyl tensor} $d^c{}_{dab}$ via the relation
\[
d^c{}_{dab} \equiv \Xi^{-1} C^c{}_{dab}.
\]

\subsubsection{A frame formalism}
Let $\{ \bme_\bma
\}$, $\bma=\bmzero,\ldots,\bmthree$ denote a $\bmg$-orthogonal frame
with associated coframe $\{ \bmomega^\bma \}$. Thus, one has that
\[
\bmg(\bme_\bma,\bme_\bmb)=\eta_{\bma\bmb}, \qquad \langle
  \bmomega^\bma,\bme_\bmb\rangle =\delta_\bmb{}^\bma.
\]
Given a vector $v^a$, its components with respect to the frame $\{ \bme_\bma
\}$ are denoted by $v^\bma$. 

Let
  $\Gamma_\bma{}^\bmc{}_\bmb$ and $\hat{\Gamma}_\bma{}^\bmc{}_\bmb$
  denote, respectively, the connection coefficients of $\nabla_\bma$
  and $\hat{\nabla}_a$ with respect to the frame $\{ \bme_\bma \}$. It
  follows then from equation \eqref{WeylToUnphysical} that
\[
\hat{\Gamma}_\bma{}^\bmc{}_\bmb = \Gamma_\bma{}^\bmc{}_\bmd + S_{\bma\bmb}{}^{\bmc\bmd}f_\bmd.
\]
In particular, one has that 
\[
f_\bma = \frac{1}{4}\hat{\Gamma}_\bma{}^\bmb{}_\bmb.
\]
Denoting by $\partial_\bma\equiv \bme_\bma{}^\mu\partial_\mu$ the
directional partial derivative in the direction of $\bme_\bma$, it
follows then that
\begin{eqnarray*}
&& \nabla_\bma T^\bmb{}_\bmc \equiv  e_\bma{}^a \omega^\bmb{}_b\omega^\bmc{}_c(\nabla_a
   T^b{}_c),\\
&&\phantom{\nabla_\bma T^\bmb{}_\bmc}=\partial_\bma T^\bmb{}_\bmc +
   \Gamma_\bma{}^\bmb{}_\bmd T^\bmd{}_\bmc -\Gamma_\bma{}^\bmd{}_\bmc T^\bmb{}_\bmd,
\end{eqnarray*}
with the natural extensions for higher rank tensors and other
covariant derivatives.

\subsection{The frame version of the extended conformal Einstein
  field equations}
In this article we will make use of a frame version of the extended
conformal Einstein field equations. In order to formulate these equations
it is convenient to define the following \emph{zero-quantities}:
\begin{subequations}
\begin{eqnarray} 
&& \hat{\Sigma}{}_\bma{}^\bmc{}_\bmb \bme_\bmc\equiv [\bme_\bma, \bme_\bmb] - (\hat{\Gamma}{}_\bma{}^\bmc{}_\bmb- \hat{\Gamma}{}_\bmb{}^\bmc{}_\bma)e_\bmc, \label{ecfe1}\\
&&\hat{\Xi}{}^\bmc{}_{\bmd\bma\bmb}\equiv \hat{P}{}^\bmc{}_{\bmd\bma\bmb} - \hat{\rho}{}^\bmc{}_{\bmd\bma\bmb}, \label{ecfe2} \\
&&\hat{\Delta}{}_{\bmc\bmd\bmb} \equiv \nabla_\bmc \hat{L}{}_{\bmd\bmb} - \nabla_\bmd
   \hat{L}{}_{\bmc\bmb} - d_\bma \hat{d}{}^\bma{}_{\bmb\bmc\bmd}, \label{ecfe3} \\
&&\Lambda{}_{\bmb\bmc\bmd} \equiv \nabla_\bma
  d^\bma{}_{\bmb\bmc\bmd} -f_\bma d^\bma{}_{\bmb \bmc \bmd}, \label{ecfe4}
\end{eqnarray}
\end{subequations}
where the components of the \emph{geometric curvature} $\hat{P}{}^\bmc{}_{\bmd\bma\bmb}$ and the
\emph{algebraic curvature} $\hat{\rho}{}^\bmc{}_{\bmd\bma\bmb}$ are given, respectively, by
\begin{eqnarray*}
&& \hat{P}{}^\bmc{}_{\bmd\bma\bmb} \equiv  \partial_\bma (\hat{\Gamma}{}_\bmb{}^\bmc{}_\bmd)- \partial_\bmb (\hat{\Gamma}{}_\bma{}^\bmc{}_\bmd) + \hat{\Gamma}{}_\bmf{}^\bmc{}_\bmd(\hat{\Gamma}{}_\bmb{}^\bmf{}_\bma - \hat{\Gamma}{}_\bma{}^\bmf{}_\bmb) + \hat{\Gamma}{}_\bmb{}^\bmf{}_\bmd \hat{\Gamma}{}_\bma{}^\bmc{}_\bmf - \hat{\Gamma}{}_\bma{}^\bmf{}_\bmd \hat{\Gamma}{}_\bmb{}^\bmc{}_\bmf,\\
&& \hat{\rho}{}^\bmc{}_{\bmd\bma\bmb} \equiv \Xi \hat{d}{}^\bmc{}_{\bmd\bma\bmb} + 2 {S}{}_{\bmd
   [\bma}{}^{\bmc\bme}\hat{ L}{}_{\bmb]\bme}, 
\end{eqnarray*}
where $\hat{L}_{\bma\bmb}$ and $d^\bmc{}_{\bmd\bma\bmb}$ denote,
respectively, the components of the Schouten tensor of
$\hat{\nabla}_a$ and the rescaled Weyl tensor with respect to the
frame $\{ \bme_\bma \}$. In terms of the zero-quantities \eqref{ecfe1}-\eqref{ecfe4}, the
\emph{extended vacuum conformal Einstein field equations} are given by
the conditions
\begin{equation}
 \hat{\Sigma}{}_\bma{}^\bmc{}_\bmb\bme_\bmc=0, \qquad \hat{\Xi}{}^\bmc{}_{\bmd\bma\bmb}=0,
 \qquad  \hat{\Delta}{}_{\bmc\bmd\bmb}=0, \qquad
 \hat{\Lambda}{}_{\bmb\bmc\bmd}=0. \label{ecfe5}
\end{equation}
In the above equations the fields $\Xi$ and $d_\bma$ ---cfr.
\eqref{Definition:CovectorD}--- are regarded as \emph{conformal gauge
  fields} which are determined by supplementary conditions. In the
present article these gauge conditions will be determined through
conformal geodesics ---see Subsection
\ref{Subsection:ConformalGeodesics} below. In order to account for
this it is convenient to define
\begin{subequations}
\begin{eqnarray}
&& \delta_\bma  \equiv d_\bma -\Xi f_\bma -\hat{\nabla}_\bma\Xi, \label{Supplementary1}\\
&& \gamma_{\bma\bmb} \equiv \hat{L}_{\bma\bmb}
   -\hat{\nabla}_\bma(\Xi^{-1} d_\bmb) -\frac{1}{2}\Xi^{-2}
   S_{\bma\bmb}{}^{\bmc\bmd}d_\bmc d_\bmd +
   \frac{1}{6}\lambda\Xi^{-2}\eta_{\bma\bmb}, \label{Supplementary2}\\
&& \varsigma_{\bma\bmb} \equiv \hat{L}_{[\bma\bmb]}
   -\hat{\nabla}_{[\bma} f_{\bmb]}. \label{Supplementary3}
\end{eqnarray}
\end{subequations}
The conditions
\begin{equation}
\delta_\bma =0, \qquad \gamma_{\bma\bmb}=0, \qquad
\varsigma_{\bma\bmb}=0,
\label{XCFESupplementary}
\end{equation}
will be called the \emph{supplementary conditions}. They play a role
in relating the Einstein field equations to the extended conformal
Einstein field equations and also in the propagation of the constraints.

\medskip
The correspondence between the Einstein field equations and the
extended conformal Einstein field equations is given by the following
---see Proposition 8.3 in \cite{CFEBook}:

\begin{proposition}
Let  
\[
(\bme_\bma, \, \hat{\Gamma}_\bma{}^\bmb{}_\bmc,
\hat{L}_{\bma\bmb},\,d^\bma{}_{\bmb\bmc\bmd})
\]
 denote a solution to
the extended conformal Einstein field equations \eqref{ecfe5} for some
choice of the conformal gauge fields $(\Xi,\,d_\bma)$ satisfying the
supplementary conditions \eqref{XCFESupplementary}. Furthermore,
suppose that
\[
 \Xi\neq 0, \qquad \det
(\eta^{\bma\bmb}\bme_\bma\otimes\bme_\bmb)\neq0
\]
 on an open subset
$\mathcal{U}$. Then the metric
\[
\tilde{\bmg}= \Xi^{-2} \eta_{\bma\bmb} \bmomega^\bma\otimes\bmomega^\bmb
\]
is a solution to the Einstein field equations \eqref{EFE} on
$\mathcal{U}$. 
\end{proposition}

\subsection{The conformal constraint equations}
\label{Subsection:ConformalConstraints}
The analysis in this article will make use of the \emph{conformal
  constraint Einstein equations}
---i.e. the intrinsic equations implied by the (standard) vacuum conformal
Einstein field equations on a spacelike hypersurface. A derivation of
these equations in its frame form can be found in \cite{CFEBook},
Section 11.4. 

\medskip
Let $\mathcal{S}$ denote a spacelike hypersurface in an unphysical
spacetime $(\mathcal{M},\bmg)$. In the following let $\{ \bme_\bma \}$ denote a $\bmg$-orthonormal
frame adapted to $\mathcal{S}$. That is, the vector $\bme_\bmzero$ is
chosen to coincide with the unit normal vector to the hypersurface and
while the spatial vectors $\{ \bme_\bmi \}$, $\bmi=\bmone, \,
\bmtwo,\, \bmthree$ are intrinsic to $\mathcal{S}$. In our signature
conventions we have that $\bmg(\bme_\bmzero,\bme_\bmzero)=-1$. The extrinsic
curvature is described by the components $\chi_{\bmi\bmj}$ of the
Weingarten tensor. One has that $\chi_{\bmi\bmj}=\chi_{\bmj\bmi}$ and,
moreover
\[
\chi_{\bmi\bmj} = -\Gamma_\bmi{}^\bmzero{}_\bmj.
\]
We denote by $\Omega$ the restriction of the spacetime conformal
factor $\Xi$ to $\mathcal{S}$ and by $\Sigma$ the normal component of
the gradient of $\Xi$. The field $l_{\bmi\bmj}$ denote the components of the
Schouten tensor of the induced metric $h_{ij}$ on $\mathcal{S}$. 

\smallskip
With the above conventions, the conformal constraint equations in the vacuum case are given by
---see \cite{CFEBook}:
\begin{subequations}
\begin{eqnarray}
&& \label{co1}  D_\bmi D_\bmj \Omega =  \Sigma \chi_{\bmi\bmj} - \Omega L_{\bmi\bmj} + s h_{\bmi\bmj}, \\
&& \label{co2} D_\bmi \Sigma= {\chi_\bmi}^\bmk D_\bmk \Omega - \Omega L_\bmi, \\
&& \label{co3} D_\bmi s=  L_\bmi \Sigma - L_{\bmi\bmk} D^\bmk \Omega, \\
&&\label{co4} D_\bmi L_{\bmj\bmk} - D_\bmj L_{\bmi\bmk}=  \Sigma d_{\bmk\bmi\bmj} + D^\bml \Omega d_{\bml\bmk\bmi\bmj}  (\chi_{\bmi\bmk} L_\bmj - \chi_{\bmj\bmk} L_\bmi), \\
&&\label{co5} D_\bmi L_\bmj - D_\bmj L_\bmi=  D^\bml \Omega d_{\bml\bmi\bmj} +{\chi_\bmi}^\bmk L_{\bmj\bmk} - {\chi_\bmj}^\bmk L_{\bmi\bmk}, \\
&&\label{co6} D^\bmk d_{\bmk\bmi\bmj}= - ({\chi^\bmk}_\bmi d_{\bmj\bmk} - {\chi^\bmk}_\bmj d_{\bmi\bmk}), \\
&&\label{co7} D^\bmi d_{\bmi\bmj}=\chi^{\bmi\bmk} d_{\bmi\bmj\bmk}, \\
&&\label{co8} \lambda= 6 \Omega s + 3 \Sigma^2 - 3 D_\bmk \Omega D^\bmk \Omega, \\
&&\label{co9} D_\bmj \chi_{\bmk\bmi} - D_\bmk \chi_{\bmj\bmi} =\Omega d_{\bmi\bmj\bmk} + h_{\bmi\bmj} L_\bmk - h_{\bmi\bmk} L_\bmj, \\
&&\label{co10}l_{\bmi\bmj}= \Omega d_{\bmi\bmj} + L_{\bmi\bmj} - \chi (
 \chi_{\bmi\bmj} - \frac{1}{4} \chi h_{\bmi\bmj} ) + \chi_{\bmk\bmi}{\chi_\bmj}^\bmk -
\frac{1}{4}\chi_{\bmk\bml}\chi^{\bmk\bml} h_{\bmi\bmj},
\end{eqnarray}
\end{subequations}
with the understanding that 
\[
h_{\bmi\bmj}\equiv g_{\bmi\bmj}=\delta_{\bmi\bmj}
\]
 and where we have defined
\[
L_\bmi\equiv L_{\bmzero\bmi}, \qquad d_{\bmi\bmj}\equiv
d_{\bmzero\bmi\bmzero\bmj}, \qquad d_{\bmi\bmj\bmk}\equiv d_{\bmi\bmzero\bmj\bmk}.
\]
The fields $d_{\bmi\bmj}$ and $d_{\bmi\bmj\bmk}$ correspond,
respectively, to the electric and magnetic parts of the rescaled Weyl
tensor. The scalar $s$ denotes the \emph{Friedrich scalar} defined as
\[
s \equiv \frac{1}{4}\nabla_a\nabla^a \Xi + \frac{1}{24}R \Xi,
\]
with $R$ the Ricci scalar of the metric $\bmg$. Finally,
$L_{\bmi\bmj}$ denote the spatial components of the Schouten tensor of
$\bmg$.

\section{Properties of the background solution}
\label{Section:BackgroundSolution}
In the following let $(\tilde{\mathcal{M}},\mathring{\tilde{\bmg}})$
denote the solution to the vacuum Einstein field equations with
positive Cosmological constant
\begin{equation}
\tilde{R}_{ab} =3 \tilde{g}_{ab}, 
\label{EFE}
\end{equation}
given by  $\tilde{\mathcal{M}}=\mathbb{R}\times\mathcal{S}$ and
\begin{equation}	
\mathring{\tilde{\bmg}} = -\mathbf{d} t \otimes \mathbf{d} t + \sinh^2 t \; \mathring{\bmgamma}
\label{BackgroundPhysicalMetric}
\end{equation}
where $\mathring{\bmgamma}$ is a (positive definite) Riemannian metric of constant negative curvature
over a compact manifold $\mathcal{S}$ such that
\[
r[\mathring{\bmgamma}]=-6.
\]
The spacetime $(\tilde{\mathcal{M}},\mathring{\tilde{\bmg}})$ is
future geodesically complete ---see Appendix \ref{Appendix:GeodesicCompleteness}.

\begin{remark}
\label{Remark:ValueLambda}
{\em The value $\lambda=3$ for the Cosmological constant is
  conventional and set for convenience. The analysis in this article
  can be carried out for any other (positive) value. Indeed, given
  $\lambda>0$ define the metric $\bar{g}_{ab}$ via the relation
\[
\bar{g}_{ab} = \frac{3}{\lambda} g_{ab}.
\] 
As this is a constant conformal rescaling, the Ricci tensor is
invariant ---i.e. $\bar{R}_{ab} =R_{ab}$; see e.g. equation  (5.6a) in
page 116 in
\cite{CFEBook}. It follows then that equation \eqref{EFE} implies
\[
\bar{R}_{ab} =\lambda \bar{g}_{ab}.
\]
}
\end{remark}

\begin{remark}
{\em The existence of compact 3-manifolds with constant negative
  scalar curvature has been analysed in the mathematical
  literature ---see \cite{Kap94}. These 3-manifolds are \emph{locally isometric} to quotients
of the hyperbolic space $\mathbb{H}^3$. The admissible topologies are
discussed in \cite{Bes08}. This class of manifolds is sometimes called
\emph{conformally rigid hyperbolic manifolds} as, despite being
conformally flat, they do not admit
globally defined conformal Killing vectors nor non-trivial tracefree
Codazzi-tensors. These properties play a crucial role in the
perturbative construction of initial data for the conformal evolution
system as discussed in Section \ref{Section:InitialData}.  }
\end{remark}

The Riemann curvature tensor $r^i{}_{jkl}[\mathring{\bmgamma}]$ of $\mathring{\bmgamma}$ is given by
\[
r_{ijkl}[\mathring{\gamma}] =
\mathring{\gamma}_{il}\mathring{\gamma}_{jk} - \mathring{\gamma}_{ik}\mathring{\gamma}_{jl}.
\]
From the above expressions it follows that
\[
\tilde{R}=12,
\]
so that 
\begin{equation}
\tilde{L}_{ab} = \frac{1}{2}\tilde{g}_{ab}
\label{PhysicalBackgroundSchouten}
\end{equation}

\medskip
In the following, a spacetime of the form given by 
$(\tilde{\mathcal{M}},\mathring{\tilde{\bmg}})$ will be known as a
\emph{background solution}. In the rest of this section we will
perform an analysis of this class of solutions to the Einstein field
equations from the point of view of conformal geometry. In particular,
we will make use of conformal geodesics to provide a \emph{canonical}
conformal extension ---see Proposition
\ref{Proposition:ConformalGeodesicsConformalFactor}.

\subsection{A class of conformal geodesics}
\label{Subsection:AClassOfCG}
In the following we will consider (metric) geodesics $x(s)$ on $(\tilde{\mathcal{M}},\mathring{\tilde{\bmg}})$ whose
tangent vector is proportional to $\bmpartial_t$ ---i.e. 
$\dot{\bmx} =\alpha \bmpartial_t$ for some proportionality function
$\alpha$ and where the overdot denotes
differentiation with respect to the affine parameter $s\in
\mathbb{R}$. The geodesic equation
\[
\tilde{\nabla}_{\dot{\bmx}} \dot{\bmx}=0
\]
implies that
\begin{eqnarray*}
&&\tilde{\nabla}_{\bmpartial_t}(\alpha\bmpartial_t)= \tilde{\nabla}_t\alpha + \tilde{\nabla}_{\bmpartial_t}\bmpartial_t\\
&&\phantom{\tilde{\nabla}_{\bmpartial_t}(\alpha\bmpartial_t)}=  \tilde{\nabla}_t\alpha + \Gamma_t{}^\mu{}_t \bmpartial_\mu.
\end{eqnarray*}
A direct calculation for the metric \eqref{BackgroundPhysicalMetric}
shows that $\Gamma_t{}^\mu{}_t=0$ so that one concludes that
$\partial_t \alpha=0$ ---that is, $\alpha$ is constant along the integral curves of
$\bmpartial_t$. Without loss of generality we then set
$\alpha=1$ so that $\bmg(\dot{\bmx},\dot{\bmx})=-1$. In summary, we
have that the curves
\[
x(t) = (t, \underline{x}_\star), \qquad \underline{x}_\star\in \mathcal{S},
\]
are non-intersecting timelike $\tilde{\bmg}$-geodesics over
$\tilde{\mathcal{M}}$. In a slight abuse of notation the coordinate
$t$ has been used as parameter of the curve.

\subsubsection{Conformal geodesics}
\label{Subsection:ConformalGeodesics}
The extended conformal Einstein field
equations are naturally suited to the use of a gauge based on conformal geodesics. 

\smallskip
A \emph{conformal geodesic} on a spacetime
$(\tilde{\mathcal{M}},\tilde{\bmg})$ is a pair
$(x(\tau),\tilde\bmbeta(\tau))$ consisting of a curve $x(\tau)$ on
$\tilde{\mathcal{M}}$, with parameter $\tau\in I\subset \mathbb{R}$,
tangent $\dot{\bmx}(\tau)$ and a covector $\tilde\bmbeta(\tau)$ along
$x(\tau)$ satisfying the equations (in index-free notation)
\begin{subequations}
\begin{eqnarray}
&& \tilde{\nabla}_{\dot{\bmx}} \dot{\bmx} = -2 \langle \tilde\bmbeta,
   \dot{\bmx}\rangle \dot{\bmx} +\tilde{\bmg}(\dot{\bmx},\dot{\bmx})
   \tilde\bmbeta^\sharp, \label{CGEqn1}\\
&& \tilde{\nabla}_{\dot{\bmx}}\tilde\bmbeta =\langle \tilde\bmbeta,\dot{\bmx}
   \rangle\tilde\bmbeta
   -\frac{1}{2}\tilde{\bmg}^\sharp(\tilde\bmbeta,\tilde\bmbeta)\dot{\bmx}^{\flat}
   +\tilde{\bmL}(\dot{\bmx},\cdot), \label{CGEqn2}
\end{eqnarray}
\end{subequations}
where $\tilde{\bmL}$ denotes the Schouten tensor of the Levi-Civita
connection $\tilde{\nabla}_a$. Associated to a conformal geodesic, it
is natural to consider a frame $\{ \bme_\bma \}$ which is \emph{Weyl
  propagated} along $x(\tau)$ according to the law
\begin{equation}
\tilde{\nabla}_{\dot{\bmx}} \bme_\bma = -\langle \tilde\bmbeta,
\bme_\bma\rangle\dot{\bmx} -\langle \tilde\bmbeta,\dot{\bmx} \rangle
\bme_\bma +\tilde{\bmg}(\bme_\bma,\dot{\bmx})\tilde\bmbeta^\sharp.
\label{CGEqn3}
\end{equation}

\subsubsection{Reparametrisation as conformal geodesic}
In the following, we will make use of the methods in the proof of Lemma 5.2 in
  \cite{CFEBook} to recast the family of geodesics discussed in
  Subsection \ref{Subsection:AClassOfCG} as conformal
  geodesics. Accordingly, we consider a reparametrisation of the form
\[
\tau \mapsto t(\tau),
\]
while we look for a 1-form $\tilde{\bmbeta}$ given by the Ansatz
\[
\tilde\bmbeta = \alpha(\tau) {\bmx}^{\prime\flat} = \alpha(t) \mathbf{d}t,
\]
where ${}^\prime$ denotes derivatives with respect to $s$. From the chain rule it follows that
\[
\dot{\bmx} = \frac{\mbox{d}t}{\mbox{d}\tau} \frac{\mbox{d}x}{\mbox{d}t} =
\dot{t} \bmx', \qquad \dot{t}\equiv \frac{\mbox{d}t}{\mbox{d}\tau}.
\]
In particular, one readily has that
\[
\tilde{\nabla}_{\dot{\bmx}} \dot{\bmx} = \dot{t}^{2}
\tilde{\nabla}_{{\bmx}'} {\bmx}' + \ddot{t} {\bmx}'.
\]
Substituting the previous expressions into equations
\eqref{CGEqn1} and
\eqref{CGEqn2}, taking into account expression
\eqref{PhysicalBackgroundSchouten} for the components of the Schouten
tensor one obtains the system of ordinary differential equations
\begin{subequations}
\begin{eqnarray}
&& \ddot{t} + \alpha \dot{t}^{2} =0, \label{CGRescaling1} \\
&& \dot{\alpha} = \frac{1}{2}\dot{t}(\alpha^2-1). \label{CGRescaling2}
\end{eqnarray}
\end{subequations}
The general solution to the above system can be found to be
\begin{eqnarray*}
&& \alpha(\tau) = c_1\tau + c_2, \\
&& t(\tau) = -2 \mbox{arctanh}(c_1 \tau +c_2) +c_3,
\end{eqnarray*}
with $c_1,\,c_2,\,c_3\in\mathbb{R}$ constants. For simplicity one can,
e.g. set $c_1=-1$, $c_2=c_3=0$ to get the simpler expressions
\begin{eqnarray*}
&& \alpha(\tau) = -\tau, \\
&&  t(\tau) = 2 \mbox{arctanh}\, \tau.
\end{eqnarray*}
Thus, observing that 
\[
\sinh\big( 2 \mbox{arctanh}\,\tau \big) = \frac{2\tau}{1-\tau^2},
\qquad \frac{\mbox{d}}{\mbox{d}\tau}\big( 2\mbox{arctanh}\,\tau \big)
=\frac{2}{1-\tau^2},
\]
it follows that the pair $(x(\tau),\tilde{\bmbeta}(\tau))$,
$\tau\in(-1,1)$ with
\[
x(\tau) = (2\,\mbox{arctanh}\,\tau, \underline{x}_\star), \qquad
\tilde{\bmbeta}(\tau) = -\frac{2\tau}{1-\tau^2}\mathbf{d}\tau,
\]
give rise to a congruence of non-intersecting conformal geodesics on
the background spacetime
$(\tilde{\mathcal{M}},\mathring{\tilde{\bmg}})$. Using the parameter $\tau$ as new coordinate in the metric
\eqref{BackgroundPhysicalMetric} one concludes that
\begin{equation}
\mathring{\tilde{\bmg}}  = \frac{4}{(1-\tau^2)^2}\bigg(
-\mathbf{d}\tau\otimes\mathbf{d}\tau +\tau^2 \mathring{\bmgamma}  \bigg).
\label{ReparametrisedMetric}
\end{equation}
Notice that the metric is singular at $\tau=\pm 1$. 

\subsubsection{The canonical factor associated to the congruence of
  conformal geodesics}
The line element
\eqref{ReparametrisedMetric} readily suggest the conformal factor
\[
\Theta \equiv \frac{1}{2}(1-\tau^2). 
\]

\begin{remark}
{\em Alternatively, we can make
use of the equation
\[
\dot{\Theta}= \langle \tilde\bmbeta, \dot{\bmx}\rangle \Theta, \qquad \langle \tilde\bmbeta, \dot{\bmx}\rangle = \alpha \dot{t} = - \frac{2\tau}{1-\tau^2}
\]
implied by the condition $\Theta^2\tilde{\bmg}(\dot{\bmx},\dot{\bmx})=-1$. 
Integrating one readily finds that 
\[
\frac{\Theta}{\Theta_\star} = \frac{1-\tau^2}{1-\tau^2_\star}
\]
where $\Theta_\star$ is the value of the conformal factor at a
fiduciary time $\tau_\star$. Observe, also, that 
\begin{eqnarray}
&&\tilde\bmbeta(\tau) = -\frac{2\tau}{1-\tau^2}\mathbf{d}\tau,
   \nonumber \\
&& \phantom{\tilde\bmbeta(\tau)} = \mathbf{d} \big( \ln \Theta(\tau) \big).\label{PhysicalBeta}
\end{eqnarray}}
\end{remark}

Following expression \eqref{ReparametrisedMetric} we introduce a new
unphysical metric $\mathring{\bmg}$ via the relation
\[
 \mathring{\bmg} = \Theta^2 \mathring{\tilde{\bmg}}, \qquad \Theta \equiv \frac{1}{2}{(1- \tau^2)}, 
\]
so as to ensure that $\Theta \geq 0$ for $|\tau|\leq 1$. It follows then that 
\begin{equation}
\mathring{\bmg}= -\mathbf{d} \tau \otimes \mathbf{d} \tau + \tau^2
\mathring{\bm{\gamma}} 
\label{BackgroundMetric}
\end{equation}
is well defined for $ \tau \in [ \tau_*, \infty)$ with
$\tau_\star>0$. For future use we define the \emph{spatial metric}
$\mathring\bmh$ 
\[
\mathring\bmh \equiv \tau^2\mathring\bmgamma, 
\]
with associated Levi-Civita connection to be denoted by $\mathring
D$. Also, denote by $\mathring{\mathfrak{D}}$ the Levi-Civita connection
 of the metric $\mathring\bmgamma$. A Penrose diagram of conformal
 representation of the background solution described by the metric
 \eqref{BackgroundMetric} is given in Figure \ref{Figure:PenroseDiagram}.

\begin{figure}[t]
\centering
\includegraphics[width=50mm]{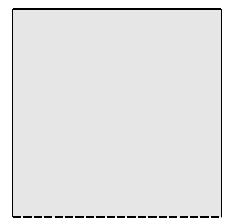}
\put(-76,140){$\mathscr{I}^+$}
\put(-76,-10){$\mathcal{S}_\star$}
\put(-150,70){$\Gamma_1$}
\put(4,70){$\Gamma_2$}
\caption{Penrose diagram of the background solution. The conformal
  representation discussed in the main text has compact sections of
  negative scalar curvature. The vertical lines $\Gamma_1$ and
  $\Gamma_2$ correspond to axes of symmetry. The solution has a singularity in the
  past and a spacelike future conformal boundary. Hence, in our discussion we only consider future evolution of
  the initial hypersurface $\mathcal{S}_\star$.}
\label{Figure:PenroseDiagram}
\end{figure}

\begin{remark}
{\em Observe that as the metrics
 $\mathring\bmgamma$ and $\mathring\bmh$ are conformally related via a
conformal factor (i.e. $\tau$) independent of the spatial coordinates, it
follows then that expressed in terms of local (spatial) coordinates
one has that
\[
\mathring D_\alpha = \mathring{\mathfrak{D}}_\alpha.
\]}
\end{remark}

\begin{remark}
\label{Remark:Geodesics}
{\em A computation readily shows that the integral curves of the
  vector field $\bmpartial_\tau$ are geodesics of the metric
  $\mathring\bmg$ given by equation \eqref{BackgroundMetric} ---that
  is, one has that
\[
\nabla_{\bmpartial_\tau}\bmpartial_\tau =0.
\]
}
\end{remark}

\begin{remark}
{\em Taking into account the expression \eqref{PhysicalBeta}, the conformal
  transformation law for conformal geodesics gives that
\[
\bmbeta = \tilde\bmbeta - \mathbf{d} \big( \ln \Theta(\tau) \big) =0.
\]
To any (non-singular) congruence of conformal geodesics one can
associate a Weyl connection $\hat\bmnabla$ via the rule
\[
\hat\bmnabla -\tilde\bmnabla =\bmS( \tilde\bmbeta).
\]
In the present case, $\tilde\bmbeta$ is a closed 1-form and, thus, the
Weyl connection is, in fact, a Levi-Civita connection which coincides
with $\bmnabla$. 
}
\end{remark}

\subsection{The background spacetime as a solution to the conformal
  Einstein field equations}
In this subsection we show how to recast the \emph{unphysical
  spacetime} $(\mathcal{M},\mathring{\bmg})$ with $\mathcal{M}=
[\tau_\star,\infty)\times \mathcal{S}$ as a solution to the conformal
Einstein field equations. This construction is conveniently done using
an adapted frame formalism. 

\subsubsection{The frame}
Let $\{\mathring{\bmc}_\bmi\}$, $\bmi=1,\, 2,\,3$, denote a
$\mathring{\bmgamma}$-orthonormal frame over $\mathcal{S}$ with associated cobasis $\{ \mathring{\bmalpha}^\bmi \}$. Accordingly, one has that 
\[
\mathring{\bm{\gamma}}(\mathring\bmc_i, \mathring\bmc_j)= \delta_{\bmi \bmj}, \qquad \langle \mathring\bmalpha^\bmj, \mathring\bmc_\bmi \rangle= \delta_\bmi{}^\bmj,
\]
so that
\[
\mathring{\bmgamma} = \delta_{\bmi\bmj} \mathring\bmalpha^\bmi \otimes \mathring\bmalpha^\bmj.
\] 
The above frame is used to introduce a $\mathring{\bmg}$-orthonormal
frame $\{ \mathring{\bme}_\bma \}$ with associated cobasis $\{
\mathring{\bmomega}^\bmb\}$ so that $\langle \mathring{\bmomega}^\bmb,
\mathring{\bme}_\bma\rangle =\delta_\bma{}^\bmb$. We do this by
setting
\begin{eqnarray*}
\mathring{\bme}_0 \equiv\bmpartial_\tau, &&
                                       \mathring{\bme}_\bmi\equiv \frac{1}{\tau}\mathring{\bmc}_\bmi,\\
\mathring{\bmomega}^0 \equiv \mathbf{d}\tau, &&
                                                \mathring{\bmomega}^\bmi
                                                = \tau \mathring{\bmalpha}^\bmi, 
\end{eqnarray*}
so that 
\[
\mathring{\bmg} = \eta_{\bma\bmb} \mathring{\bmomega}^\bma \otimes \mathring{\bmomega}^\bmb.
\]

\begin{remark}
{\em It follows that all the coefficients of the frame are smooth
  ($C^\infty$) over $[\tau_\star,\infty)\times\mathcal{S}$, $\tau_\star>0$.}
\end{remark}

\subsubsection{The connection coefficients}
The connection coefficients $\mathring\gamma_\bmi{}^\bmk{}_\bmj$ of
the Levi-Civita connection $\mathring{D}$ with respect to the frame
$\{\mathring\bmc_\bmi \}$ are defined through the relations
\[
\mathring D_\bmi \mathring\bmc_\bmj =\mathring\gamma_\bmi{}^\bmk{}_\bmj
\mathring\bmc_\bmk, \qquad \gamma_\bmi{}^\bmk{}_\bmj \equiv \langle
\mathring\bmalpha^\bmk, \mathring D_\bmi \mathring\bmc_\bmj \rangle.
\]
Similarly, for the connection coefficients
$\mathring\Gamma_\bmi{}^\bmk{}_\bmj$ of the Levi-Civita connection
$\mathring\nabla$ with respect to the frame $\{ \mathring\bme_\bma \}$
one has that 
\[
\mathring \nabla_\bma \mathring\bme_\bmb =
\mathring\Gamma_\bma{}^\bmc{}_\bmb \mathring\bme_\bmc, \qquad
\mathring\Gamma_\bma{}^\bmc{}_\bmb \equiv \langle \mathring\bmomega^\bmc, \mathring \nabla_\bma \mathring\bme_\bmb\rangle.
\]
We now proceed to compute the various connection coefficients.

\medskip
\noindent
\textbf{The coefficients $\mathring\Gamma_\bmi{}^\bmk{}_\bmj$.}
Recalling the definition the connection coefficients and of the basis
fields $\{\mathring{\bme}_i\}$ and $\{ \mathring\bmomega^j \}$ one has
that 
\begin{eqnarray*}
&& \mathring\Gamma_\bmi{}^\bmk{}_\bmj = \langle
  \bmomega^k,\mathring\nabla_\bmi \mathring\bme_\bmj \rangle =\langle
    \bmomega^\bmk, \mathring e_\bmi{}^\alpha \mathring\nabla_\alpha
  \mathring\bme_\bmj\rangle
 \\
&& \phantom{\mathring\Gamma_\bmi{}^\bmk{}_\bmj} = \frac{1}{\tau}
   \langle \mathring\bmalpha^\bmk, \mathring c_\bmi{}^\alpha \mathring\nabla_\alpha
   \mathring \bmc_\bmj\rangle  = \frac{1}{\tau}
   \langle \mathring\bmalpha^\bmk, \mathring c_\bmi{}^\alpha \mathring{D}_\alpha
   \mathring \bmc_\bmj\rangle = \frac{1}{\tau}\langle
   \mathring\bmalpha^\bmk, \mathring{D}_\bmi \mathring\bmc_\bmj
   \rangle\\
&& \phantom{\mathring\Gamma_\bmi{}^\bmk{}_\bmj} =\frac{1}{\tau}\mathring\gamma_\bmi{}^\bmk{}_\bmj.
\end{eqnarray*}

\medskip
\noindent
\textbf{The coefficients $\mathring{\Gamma}_\bmzero{}^\bma{}_\bmzero$.} Recall
that $\mathring\bme_\bmzero =\bmpartial_\tau$ is tangent to geodesics ---see Remark
\ref{Remark:Geodesics}. Thus,
\[
\mathring\nabla_\bmzero \mathring\bme_\bmzero = \mathring\Gamma_\bmzero{}^\bmc{}_\bmzero \mathring\bme_\bmc,
\]
from where it follows that 
\[
\mathring\Gamma_\bmzero{}^\bma{}_\bmzero =0.
\]

\medskip
\noindent
\textbf{The coefficients $\mathring{\Gamma}_\bmi{}^\bmj{}_\bmzero$ and $\mathring{\Gamma}_\bmi{}^\bmzero{}_\bmj$.} In this
case we have that 
\[
\mathring\Gamma_\bmi{}^\bmj{}_\bmzero =\langle \mathring\bmomega^\bmj,
\mathring\nabla_\bmi \mathring\bme_\bmzero \rangle =\mathring\chi_\bmi{}^\bmj,
\]
where $\chi_\bmi{}^\bmj$ denote the components of the \emph{Weingarten
  tensor}. Defining $\mathring\chi_{\bmi\bmj} \equiv \eta_{\bmj\bmk}
\mathring\chi_\bmi{}^\bmk$, one has that $\mathring\chi_{\bmi\bmj} =\mathring\chi_{(\bmi\bmj)}$
as the congruence defined by $\bmpartial_\tau$ can readily be verified
to be hypersurface orthogonal. Thus, in this case $\mathring\chi_{\bmi\bmj}$
coincides with the components of the extrinsic curvature of the
hypersurfaces of constant $\tau$. To compute $\mathring{\chi}_{\bmi\bmj}$ recall
that
\[
\chi_{ab} = -\frac{1}{2} \mathcal{L} _{\bmpartial_\tau} h_{ab},
\]
where $ \mathcal{L} _{\bmpartial_\tau}$ denotes the Lie derivative
along the direction of $\bmpartial_t$. As
\[
\mathcal{L} _{\bmpartial_\tau} \mathring\bmh = \frac{1}{2} \mathcal{L}
_{\bmpartial_\tau} \big(\tau^2 \mathring \bmgamma\big) = 2\tau
\mathring\bmgamma = \frac{2}{\tau} \mathring\bmh,
\]
one concludes that 
\[
\mathring\chi_{\bmi\bmj} =-\frac{1}{\tau} \delta_{\bmi\bmj}. 
\]
Exploiting the metricity of the connection $\mathring\nabla$ one finds
that, moreover,
\[
\mathring{\Gamma}_\bmi{}^\bmzero{}_\bmj = - \mathring{\chi}_{\bmi\bmj} =
\frac{1}{\tau} \delta_{\bmi\bmj}. 
\]

\medskip
\noindent
\textbf{The coefficients $\mathring\Gamma_\bmzero{}^\bmj{}_\bmi$.} In this
case one readily finds that 
\begin{eqnarray*}
&& \mathring\Gamma_\bmzero{}^\bmj{}_\bmi = \langle
   \mathring\bmomega^\bmj,\mathring\nabla_\bmzero \mathring\bme_\bmi\rangle=
   \langle \mathring\bmomega^\bmj ,\mathring\nabla_\bmzero \bigg(
   \frac{1}{\tau}\mathring\bmc_\bmi\bigg) \rangle\\
&& \phantom{\mathring\Gamma_\bmzero{}^\bmj{}_\bmi} = -\frac{1}{\tau}\langle
   \mathring\bmalpha^\bmj,\mathring\bmc_\bmi \rangle=
   -\frac{1}{\tau^2}\langle \tau \mathring\bmalpha^\bmj, \mathring\bmc_\bmi\rangle\\
&& \phantom{\mathring\Gamma_\bmzero{}^\bmj{}_\bmi} = -\frac{1}{\tau}\delta_\bmi{}^\bmj.
\end{eqnarray*}

\medskip
\noindent
\textbf{The coefficients $\mathring\Gamma_\bmzero{}^\bmzero{}_\bmi$.} In this
case, one readily finds that 
\[
\mathring\Gamma_\bmzero{}^\bmzero{}_\bmi =\langle
\mathring\bmomega^\bmzero,\mathring\nabla_\bmzero \mathring\bme_\bmi \rangle =
\langle \mathring\bmomega^\bmzero,\mathring\nabla_\bmzero \bigg(
\frac{1}{\tau}\mathring\bmc_\bmi \bigg) \rangle= -\frac{1}{\tau^2}\langle \mathbf{d}\tau,\mathring\bmc_\bmi \rangle=0. 
\]

\medskip
\noindent
\textbf{The coefficients $\mathring\Gamma_\bmi{}^\bmzero{}_\bmzero$.} Observing
that $[\mathring\bme_i,\mathring\bme_\bmzero]=0$ and recalling that in the absence of torsion one has that
\[
[\mathring\bme_\bmi,\mathring\bme_\bmzero] =\bigg( \mathring\Gamma_\bmi{}^\bmc{}_\bmzero-\mathring\Gamma_\bmzero{}^\bmc{}_\bmi \bigg)\bme_\bmc,
\]
it follows from the previous results that
\[
\mathring\Gamma_\bmi{}^\bmzero{}_\bmzero=0.
\]

\begin{remark}
{\em It follows that all the coefficients of the connection are smooth
  ($C^\infty$) over $[\tau_\star,\infty)\times\mathcal{S}$.}
\end{remark}

\begin{remark}
{\em For latter use it is observed that the extrinsic
  curvature (Weingarten tensor) can be written in abstract index
  notation as
\begin{equation}
\mathring{\chi}_{ij} = -\frac{1}{\tau}\mathring{h}_{ij}.
\label{WeingartenBackground}
\end{equation}
}
\end{remark}

\subsubsection{Conformal fields}
The next step is the computation of the components of the conformal
fields appearing in the extended conformal Einstein field
equations. To this end, we make use of the conformal Einstein
constraints discussed in Section \ref{Subsection:ConformalConstraints}.

\medskip
We make use of an adapted frame with
$\bme_0=\bmpartial_\tau$ and make the identification $\Omega\mapsto
\Theta$ in equations \eqref{co1}-\eqref{co10}. Observe that one has that 
\[
\mathring{D}_i \Omega =0.
\]

\medskip
\noindent
\textbf{The scalars $\Sigma$ and $s$.} By definition one has that 
\[
\mathring{\Sigma} \equiv \bmn (\Theta) = -\partial_\tau \Theta = \tau.
\]
The minus sign arises from the fact that in our conventions
$(\mathbf{d}\tau)^\sharp =-\bmpartial_\tau$.  Using the later in the
conformal equation \eqref{co8} with $\lambda=3$ one readily concludes
\[
\mathring{s}=1.
\]

\medskip
\noindent
\textbf{Components of the Schouten tensor.} The constraint \eqref{co2}
readily yields for $\Theta\geq 0$ that
\[
\mathring{L}_i=0.
\]
The spatial components, $\mathring{L}_{ij}$, are computed using the constraint
\eqref{co10}. Observing that in our case $\mathring{D}_i\mathring{D}_j \Theta=0$ one readily
concludes that
\[
\mathring{L}_{ij}=0.
\]
Thus, all the components of the Schouten tensor, except for its trace,
vanish. This trace is proportional to the Ricci scalar of the metric
\eqref{BackgroundMetric}. 

\medskip
\noindent
\textbf{Components of the rescaled Weyl tensor.} The constraint \eqref{co9}
offers an easy way of computing the magnetic part of the rescaled Weyl
tensor. As $\mathring{D}_j \mathring{\chi}_{ki} = 0$ and we already know that $\mathring{L}_i=0$, it follows then that $\mathring{d}_{ijk} =0$ so
that, in fact, 
\[
 \mathring{d}{}^*{}_{ij}=0.
\]
To compute the electric part of the rescaled Weyl tensor we make use
of the constraint equation \eqref{co10}. This equation requires
knowing the value of the Schouten tensor $\mathring{l}_{ij}$ of the metric
$\mathring\bmh$. From the definition of the 3-dimensional Schouten
tensor one readily finds that if $r[\mathring\bmgamma]=-6$, then 
\[
\mbox{\textbf{\em Schouten}}[\mathring\bmgamma] = -\frac{1}{2}\mathring\bmgamma.
\]
Now, we have that $\mathring\bmh =\tau^2 \mathring\bmgamma$ so that
$\mathring\bmh$ and $\mathring\bmgamma$ are conformally
related. However, the conformal factor does not depend on the spatial
coordinates. It follows then, from the conformal transformation rule
of the Schouten tensor that
\[
\mbox{\textbf{\em Schouten}}[\mathring\bmgamma] = \mbox{\textbf{\em Schouten}}[\mathring\bmh].
\]
Hence, one has that 
\[
\mathring{l}_{ij} = -\frac{1}{2}\mathring\gamma_{ij} = -\frac{1}{\tau^2}\mathring
h_{ij}. 
\]
Now, a calculation using equation \eqref{WeingartenBackground} reveals
that
\[
\mathring{l}_{ij} = - \mathring{\chi} \big(
  \mathring{\chi}_{ij} - \frac{1}{4} \mathring{\chi} \mathring{h}_{ij} \big) + \mathring{\chi}_{ki}\mathring{\chi}{}_j{}^k -
  \frac{1}{4}\mathring{\chi}_{kl}\mathring{\chi}^{kl} \mathring{h}_{ij}
\] 
so that 
\[
\mathring{d}_{ij}=0.
\]

\begin{remark}
{\em In summary, one has that the metric \eqref{BackgroundMetric} is
conformally flat.}
\end{remark}

\medskip
\noindent
\textbf{Ricci scalar.} Finally, although it does not appear as an
unknown in the extended conformal Einstein equations, it is of
interest to compute the Ricci scalar of the metric. To do this we
observe that from the definition of the Friedrich scalar one has that
\[
\mathring{R}\Theta = 24 \big( s-\frac{1}{4}\mathring{\nabla}_c\mathring{\nabla}^c\Theta\big).
\] 
A computation readily yields $\mathring{\nabla}_c\mathring{\nabla}^c\Theta=-2$ so that one
concludes that
\[
\mathring{R} = \frac{72}{1-\tau^2}.
\]
That is, the Ricci scalar is singular at $\tau=1$.

\begin{remark}
{\em Although the Ricci scalar of the background solution solution is
  singular, this will not pose any difficulty in our subsequent
  analysis as the Ricci scalar does not appear as an unknown in the
  extended conformal Einstein field equations.}
\end{remark}

\section{Evolution equations}
\label{Section:EvolutionEqns}
In this section we discuss the evolution system associated to the
extended conformal Einstein equations \eqref{ecfe5} written in terms
of a conformal Gaussian system. This evolution system
is central in the discussion of the stability of the background
spacetime. In addition, we also
discuss the subsidiary evolution system satisfied by the
zero-quantities associated to the field equations, \eqref{ecfe1}-\eqref{ecfe4},
and the supplementary zero-quantities \eqref{Supplementary1}-\eqref{Supplementary3}. The
subsidiary system is key in the analysis of the so-called \emph{propagation of the
constraints} which allows to establish the relation between a solution
to the extended conformal Einstein equations \eqref{ecfe5} and the Einstein
field equations \eqref{EFE}. 

\subsection{The conformal Gaussian gauge}
In order to obtain suitable evolution equations for the conformal
fields, we make use of a \emph{conformal Gaussian gauge}. More
precisely, we assume that we are working on a region
$\mathcal{U}\subset\mathcal{M}$ which can be covered by a congruence
of non-intersecting conformal geodesics. Choosing 
\[
\Theta_\star =\frac{1}{2}, \qquad \dot{\Theta}_\star =0, \qquad
\ddot{\Theta}_\star =-\frac{1}{2},
\]
for $\tau=\tau_\star$, $\tau_\star \in(0,1)$, the following
Proposition gives a conformal factor associated to the curves of the
congruence ---see e.g. \cite{CFEBook}, Proposition 5.1, page 133:

\begin{proposition}
\label{Proposition:ConformalGeodesicsConformalFactor}
Let $(\tilde{\mathcal{M}},\tilde{\bmg})$ denote an Einstein
spacetime. Suppose that $(x(\tau),\bmbeta(\tau))$ is a solution to the
conformal geodesic equations \eqref{CGEqn1}-\eqref{CGEqn2} and that
$\{ \bme_\bma \}$ is a $\bmg$-orthogonal frame propagated along the
curve according to \eqref{CGEqn3}. If 
\[
\bmg = \Theta^2 \tilde{\bmg}, \qquad \mbox{such that}, \quad \bmg(\dot{\bmx},\dot{\bmx})=-1,
\]
then the conformal factor $\Theta$ satisfies
\[
\Theta(\tau) =\Theta_\star +\dot{\Theta}_\star(\tau-\tau_\star)
+\frac{1}{2}\ddot{\Theta}_\star (\tau-\tau_\star)^2,
\]
where the coefficients $\Theta_\star$, $\dot{\Theta}_\star$,
$\ddot{\Theta}_\star$ are constant along the conformal geodesic and
are subject to the constraints
\[
\dot{\Theta}_\star = \langle \bmbeta_\star,\dot{x}_\star
\rangle\Theta_\star, \qquad \Theta_\star\ddot{\Theta}_\star =
\frac{1}{2}\tilde{\bmg}^\sharp(\bmbeta_\star,\bmbeta_\star) + \frac{1}{6}\lambda.
\]
Furthermore, along each conformal geodesic one has
\[
\Theta\bmbeta_\bmzero =\dot{\Theta}_\star, \qquad \Theta \beta_\bmi =\Theta_\star\bmbeta_{\bmi\star}.
\]
\end{proposition}

\begin{remark}
{\em Thus, if one has a congruence of non-intersecting conformal
  geodesics in a region $\mathcal{U}$ of spacetime, then the above
  proposition provides a \emph{canonical way} of obtaining a conformal
  extension. This strategy naturally leads to a so-called
  \emph{conformal Gaussian gauge}. }
\end{remark}

The Proposition \ref{Proposition:ConformalGeodesicsConformalFactor} gives the
conformal factor
\begin{equation}
\Theta(\tau) =\frac{1}{2}\big( 1-(\tau-\tau_\star)^2\big)
\label{UniversalCF}
\end{equation}
along the curves of the congruences. The choice of initial data for the
conformal factor is associated to a congruence that leaves
orthogonally a fiduciary initial hypersurface $\mathcal{S}_\star$ with
$\tau=\tau_\star$ ---notice, however, that the congruence of conformal
geodesics is, in general, not hypersurface orthogonal. 

\begin{remark} 
{\em Since the conformal factor $\Theta$ given by equation
  \eqref{UniversalCF}  does not depend on the initial data for the
  evolution equations it can be regarded as universal ---i.e. valid
  not only for the background solution but also for perturbations
  thereof. Similarly, a consequence of Proposition 2, it follows that
  the components ${d}_\bma$ of the the covector $\bmd$ are, in the
  same sense, universal.}
\end{remark}

\smallskip
Along the congruence of conformal geodesics one considers a
$\bmg$-orthogonal frame $\{\bme_\bmzero\}$ which is Weyl-propagated
and such that $\bmtau=\bme_\bmzero$. The Weyl connection $\hat{\nabla}_a$ associated to
the congruence then satisfies
\[
\hat{\nabla}_\bmtau \bme_\bma=0, \qquad \hat{\bmL}(\tau,\cdot)=0,
\]
which is equivalent to
\[
\hat{\Gamma}_\bmzero{}^\bmb{}_\bmc=0, \qquad f_\bmzero=0, \qquad \hat{L}_{\bmzero\bma}=0,
\]
---see e.g. \cite{CFEBook}, Section 13.4, page 366. By choosing the parameter, $\tau$, of the conformal geodesics as time
coordinate one gets the additional gauge condition 
\[
\bme_\bmzero = \bmpartial_\tau, \qquad e_\bmzero{}^\mu =\delta_0{}^\mu.
\]
On $\mathcal{S}_\star$ we choose some local coordinates
$\underline{x}=(x^\alpha)$. Assuming that each curve of the congruence
intersects $\mathcal{S}_\star$ only once, one can extend the
coordinates off the initial hypersurface by requiring them to be
constant along the conformal geodesic which intersects
$\mathcal{S}_\star$ at the point with coordinates $\underline{x}$. The
coordinates $(\tau,\underline{x})$ thus obtained are known as
\emph{conformal Gaussian coordinates}. 

\subsection{The main evolution system}
One of the main advantages of writing the conformal field equations in
terms of zero-quantities and using a frame formalism is that the
various evolution equations can be readily identified as certain
components of the zero-quantities. 

\medskip
The required evolution equations for the frame components, connection coefficients and components of the Schouten tensor are obtained from the conditions
\begin{equation}
\hat{\Sigma}{}_{\bmzero}{}^\bmc{}_\bmb \bme_\bmc=0, \qquad
 \hat{\Xi}{}^\bmc{}_{\bmd\bmzero\bmb}=0, \qquad  \hat{\Delta}{}_{\bmzero\bmb\bmc}=0.  \label{ecfe6}
\end{equation}
In particular, the evolution equation for components of the covector
$f_a$ defining the Weyl connection is given by
\[
\hat{\Xi}{}^\bmc{}_{\bmc\bmzero\bmb}=0. 
\]
In the following we analyse each of these equations in more detail.

\subsubsection{Evolution equations for the components of the frame}
Now, starting from equation \eqref{ecfe1}
\[
 \hat{\Sigma}{}_\bma{}^\bmc{}_\bmb \bme_\bmc \equiv [\bme_\bma, \bme_\bmb] -
 (\hat{\Gamma}{}_\bma{}^\bmc{}_\bmb-\hat{\Gamma}{}_\bmb{}^\bmc{}_\bma)\bme_\bmc 
\]
and writing $\bme_\bma = e_\bma{}^\mu\bmpartial_\mu$, it follows that
the condition $\hat{\Sigma}{}_\bma{}^\bmc{}_\bmb\bme_\bmc=0$ implies
\[
(\partial_\bma e_\bmb{}^\nu - \partial_\bmb e_\bma{}^\nu)
=(\hat{\Gamma}{}_\bma{}^\bmc{}_\bmb-\hat{\Gamma}{}_\bmb{}^\bmc{}_\bma)e_\bmc{}^\nu,
\qquad \partial_\bma \equiv e_\bma{}^\mu\partial_\mu.
\]
Setting $\bma=\bmzero$ it follows that the evolution equation for the
components of the frame takes the form
\begin{equation}
\partial_\bmzero e_\bmb{}^\nu = - \hat{\Gamma}{}_\bmb{}^\bmc{}_\bmzero e_\bmc{}^\nu. 
\label{ev1}
\end{equation}

\subsubsection{Evolution equations for the components of the connection}
In order to obtain the evolution equation for the components of the
frame not determined by the gauge conditions one considers the
condition $\hat{\Xi}{}^\bmc{}_{\bmd\bmzero\bmb}=0$.

\smallskip
Now, since 
\[
\hat{P}{}^\bmc{}_{\bmd\bmzero\bmb}=\partial_\bmzero(\hat{\Gamma}{}_\bmb{}^\bmc{}_\bmd)-\partial_\bmb(\hat{\Gamma}{}_\bmzero{}^\bmc{}_\bmd)+(\hat{\Gamma}{}_\bmb{}^\bmf{}_\bmd
\hat{\Gamma}{}_\bmzero{}^\bmc{}_\bmf- \hat{\Gamma}{}_\bmzero{}^\bmf{}_\bmd
\hat{\Gamma}{}_\bmb{}^\bmc{}_\bmf) +\hat{\Gamma}{}_\bmf{}^\bmc{}_\bmd(
\hat{\Gamma}{}_\bmb{}^\bmf{}_\bmzero - \hat{\Gamma}{}_\bmzero{}^\bmf{}_\bmb), 
\]
then using the gauge condition
$\hat{\Gamma}{}_\bmzero{}^\bmc{}_\bmd=0$ one has that 
\[
\hat{P}{}^\bmc_{\bmd\bmzero\bmb}=\bme_\bmzero(\hat{\Gamma}{}_\bmb{}^\bmc{}_\bmd)
+\hat{\Gamma}{}_\bmf{}^\bmc{}_\bmd \hat{\Gamma}{}_\bmb{}^\bmf{}_\bmzero.
\]
In addition, observing that 
\[
S_{\bmd[\bmzero}{}^{\bmc\bme}\hat{L}{}_{\bmb]\bme}=\delta_\bmd{}^\bmc \hat{L}{}_{\bmb\bmzero}+\delta_\bmzero{}^c\hat{L}{}_{\bmb\bmd}- g_{\bmd\bmzero}g^{\bmc\bme}\hat{L}{}_{\bmb\bme}-\delta_\bmd{}^\bmc \hat{L}{}_{\bmzero\bmb}-
\delta_\bmb{}^\bmc \hat{L}{}_{\bmb\bmzero}+ g_{\bmd\bmb}g^{\bmc\bme}\hat{L}{}_{\bmzero\bme},
\]
together with the gauge condition $\hat{L}_{\bmzero\bma}=0$, it
follows that
\[
\hat{\rho}{}^\bmc{}_{\bmd\bmzero\bmb}=\Theta \hat{d}{}^\bmc{}_{\bmd\bmzero\bmb}+2\delta_\bmd{}^\bmc\hat{L}{}_{\bmb\bmzero}+2\delta_\bmzero{}^\bmc \hat{L}{}_{\bmb\bmd}-2
\eta_{\bmd\bmzero}\eta^{\bmc\bme}\hat{L}{}_{\bmb\bme},
\]
where it has been used that
$g_{\bmd\bmzero}g^{\bmc\bme}=\eta_{\bmd\bmzero}\eta^{\bmc\bme}$. It
follows that the evolution equation for the coefficients of the
connection not determined by the gauge is given by
\[
\partial_\bmzero(\hat{\Gamma}{}_\bmb{}^\bmc{}_\bmd)
  +\hat{\Gamma}{}_\bmf{}^\bmc{}_\bmd \hat{\Gamma}{}_\bmb{}^\bmf{}_\bmzero=2
  \eta_{\bmd\bmzero}\eta^{\bmc\bme}\hat{L}{}_{\bmb\bme} -2\delta_\bmd{}^\bmc \hat{L}{}_{\bmb\bmzero}-2\delta_\bmzero{}^\bmc\hat{L}{}_{\bmb\bmd}-\Theta
  \hat{d}{}^\bmc{}_{\bmd\bmzero\bmb}. 
\]
The above expression can be written in terms of the Levi-Civita
connection coefficients $\Gamma_\bma{}^\bmb{}_\bmc$ and the 1-form
$f_\bma$ through the relation
\[
\hat{\Gamma}{}_\bma{}^\bmb{}_\bmc= \Gamma_\bma{}^\bmb{}_\bmc+ S_{\bma\bmb}{}^{\bmc\bmd}f_\bmd. 
\]
In particular, since
\[ 
f_\bma=
  \frac{1}{4}\hat{\Gamma}{}_\bma{}^\bmb{}_\bmb, 
\]
it follows from the gauge condition $f_\bmzero=0$ and $\hat{\Xi}{}^\bmc{}_{\bmc\bmzero\bmb}=0$ that
\[
\partial_\bmzero f_\bmi + f_\bmj \hat{\Gamma}{}_\bmi{}^\bmj{}_\bmzero=
\hat{L}{}_{\bmi\bmzero}. 
%\label{ev3}
\]

\subsubsection{Evolution equations for the components of the Schouten tensor}
The evolution equations for the components of the Schouten tensor not determined by the gauge
 are obtained from the condition $\hat{\Delta}{}_{\bmzero \bmd\bmb}=0$. It
 follows then that 
\[
\nabla_\bmzero \hat{L}{}_{\bmd\bmb} - \nabla_\bmd
\hat{L}{}_{\bm0\bmb} - d_\bma d^\bma{}_{\bmb\bmzero\bmd}=0.
\]
However, in the conformal Gaussian gauge one has that
$\hat{L}{}_{\bmzero\bmb}=0$ so that the evolution equation for the components
of the Schouten tensor can be simplified to
\[
\partial_\bmzero \hat{L}{}_{\bmd\bmb}=\hat{\Gamma}{}_\bmzero{}^\bmc{}_\bmd
  \hat{ L}{}_{\bmc\bmb}+ \hat{\Gamma}{}_\bmzero{}^\bmc{}_\bmb \hat{L}{}_{\bmd\bmc} + d_\bma
  d{}^\bma{}_{\bmb\bmzero\bmd}=0, 
%\label{ev2} 
\]
as $\hat{\Gamma}{}_\bmzero{}^\bmc{}_\bmd =0$.

\subsubsection{Evolution equations for the components of the rescaled
  Weyl tensor}
The evolution equations for the components of the Weyl tensor are
extracted from the decomposition of the zero-quantity
$\Lambda_{\bmb\bmc\bmd}$. As this zero-quantity contains a contracted
derivative, the decomposition is more involved than for the other
zero-quantities. As in the case of the conformal constraint
equations, this analysis is best done using the decomposition of the
rescaled Weyl tensor in its electric and magnetic parts with respect
to the tangent to the congruence of conformal geodesics on which our
gauge is based. 

\medskip
In the following, let $h_a{}^b$ denote the projector to the
hyperplanes orthogonal to the tangent vector field $\tau^a$ to the
congruence of conformal geodesics. One has that 
\[
h_a{}^b=\delta_a{}^b - \tau_a \tau^b,
\]
so that 
\[
\begin{split} 
\Lambda_{\bmb\bmc\bmd}&=\nabla^\bma (\delta_\bma{}^\bmf d_{\bmf\bmb\bmc\bmd})=\delta_\bma{}^\bmf \nabla^\bma  d_{\bmf\bmb\bmc\bmd}\\
&= \tau^\bmf \tau_\bma \nabla^\bma  d_{\bmf\bmb\bmc\bmd}+h_\bma{}^\bmf\nabla^\bma  d_{\bmf\bmb\bmc\bmd}\\
& = \mathcal{D}  d_{\bmf\bmb\bmc\bmd} +\mathcal{D}^\bmf
d_{\bmf\bmb\bmc\bmd} \tau^\bmf,
\end{split}
\]
where $\mathcal{D}_a\equiv h_a{}^b\nabla_b$ and $\mathcal{D}\equiv
\tau^a\nabla_a$ denote, respectively, the Sen and Fermi covariant
derivatives associated to the congruence.  Now, observing that the acceleration and Weingarten
tensor of the congruence are given, respectively by 
\begin{eqnarray*}
&& a_a \equiv \tau^b \nabla_b \tau_a= \mathcal{D} \tau_a, \\
&& \chi_{ab}  \equiv h_a{}^c \nabla_c \tau_b =\mathcal{D}_a \tau_b,  
\end{eqnarray*}
it follows that
\[
\begin{split}
\Lambda_{\bmb\bmc\bmd}\tau^\bmc &=\Lambda_{\bmb\bmzero\bmd} =\tau^\bmc
\mathcal{D}(\tau^\bmf d_{\bmf\bmb\bmc\bmd})+ \tau^\bmc
\mathcal{D}^\bmf d_{\bmf\bmb\bmc\bmd}-a^\bmf \tau^\bmc d_{\bmf\bmb\bmc\bmd}\\
&=\mathcal{D}(\tau^\bmf \tau^\bmc d_{\bmf\bmb\bmc\bmd})+\mathcal{D}^\bmf d_{\bmf\bmb\bmzero\bmd} - a^\bmf d_{\bmf\bmb\bmzero\bmd}-
a^\bmc d_{\bmzero\bmb\bmc\bmd}- \chi^{\bmf\bmc} d_{\bmf\bmb\bmc\bmd}, 
\end{split}
\]
so that
\[
 \Lambda_{\bmb \bmzero \bmd}=\mathcal{D} d_{\bmzero \bmb \bmzero \bmd} + \mathcal{D}^\bmf d_{\bmf\bmb\bmzero\bmd} -
 a^\bmf d_{\bmf\bmb \bmzero \bmd}- a^\bmc d_{\bmzero \bmb\bmc\bmd} - \chi^{\bmf\bmc}d_{\bmf\bmb\bmc\bmd}. 
\]
To further simplify we make use of the decomposition
\[
d_{\bma\bmb\bmc\bmd}=2(l_{\bmb[\bmc} d_{\bmd]\bma}-l_{\bma[\bmc}
d_{\bmd]\bmb}) - 2(\tau_{[\bmc}{d^*}_{\bmd]\bme}
{\epsilon^\bme}_{\bma\bmb}+
\tau_{[\bma}{d^*}_{\bmb]\bme}{\epsilon^\bme}_{\bmc\bmd}), 
\]
of the rescaled Weyl tensor in terms of its electric part
$d_{\bma\bmb}$ and magnetic part $d^*_{\bma\bmb}$ with respect to the
vector field $\tau^a$ where  $l_{ab}=h_{ab}-\tau_a \tau_b$ to obtain
\[
\begin{split}
\Lambda_{\bmb \bmzero \bmd}=& \mathcal{D} d_{\bmb \bmd} +
\mathcal{D}^\bmf d_{\bmf\bmb \bmd} - a^\bmf d_{\bmf\bmb \bmd}- a^\bmc d_{ \bmb\bmc\bmd} - 2\chi^{\bmf\bmc}(l_{\bmb[\bmc} d_{\bmd]\bmf}-l_{\bmf[\bmc} d_{\bmd]\bmb}) \\
&+ 2\chi^{\bmf\bmc}(\tau_{[\bmc}{d^*}_{\bmd]\bme} {\epsilon^\bme}_{\bmf\bmb}+
\tau_{[\bmf}{d^*}_{\bmb]\bme}{\epsilon^\bme}_{\bmc\bmd}). 
\end{split}
\]
To finally extract the required evolution equation we consider
$\Lambda_{(\bmb|\bmzero|\bmd)}$. Observing that all the involved
tensors are spatial one obtains, after some simplification, that 
\begin{equation}
\Lambda_{(\bmi| \bmzero| \bmj)}=   \partial_\bmzero
  d_{\bmi\bmj}+{\epsilon^{\bmk\bml}}_{(\bmi}D_{|\bml|} d^*{}_{\bmj)\bmk}-
  2a_\bml{\epsilon^{\bmk\bml}}_{(\bmi}d^*{}_{\bmj)\bmk}+ \chi d_{\bmi\bmj} - 2
  {\chi^\bmk}_{(\bmi}d_{\bmj)\bmk}=0. \label{EvolutionElectricWeyl}
\end{equation}

\smallskip
To complete the system of evolution equations for the components of
the Weyl tensor one carries out a completely analogous
calculation with the zero-quantity 
\[
\Lambda^*_{\bmb\bmc\bmd}\equiv \nabla^\bma {d^*}_{\bma\bmb\bmc\bmd} 
\]
and the decomposition
\[
d^*_{\bma\bmb\bmc\bmd}= 2 (l_{\bmb[\bmc}{d^*}_{\bmd]\bma}- l_{\bmf[\bmc}{d^*}_{\bmd]\bmb}) + 2(
\tau_{[\bmc}{d}_{\bmd]\bme} {\epsilon^\bme}_{\bma\bmb}+ \tau_{[\bma}{d}_{\bmb]\bme}
\epsilon^\bme{}_{\bmc\bmd}),
\] 
where the Hodge dual of the rescaled Weyl tensor is defined as
\[
d^*_{abcd} \equiv \frac{1}{2}\epsilon_{ab}^{ef}d_{cdef}.
\]

More precisely, the decomposition 
\[
 \Lambda^*_{\bmb\bmc\bmd}=\tau^\bma \mathcal{D} d^*_{\bma\bmb\bmc\bmd}+ \mathcal{D}^\bma {d^*}_{\bma\bmb\bmc\bmd} , 
\]
leads, after a lengthy computation, to the evolution equation
\begin{equation}
\Lambda^*{}_{(\bmi| \bmzero| \bmj)}=    \partial_\bmzero d^*_{\bmi\bmj} -
  {\epsilon^\bmk}_{\bml(\bmi}D^\bml d_{\bmj)\bmk} - 2 a^\bml {\epsilon_{\bml(\bmi}}^\bmk d_{\bmj)\bmk} +
  \chi d^*_{\bmi\bmj} - 2 {\chi^\bmk}_{(\bmi}d^*_{\bmj)\bmk}=0, \label{EvolutionMagneticWeyl}
\end{equation}
in which all the fields are spatial.

\begin{remark}
{\em The zero-quantities $\Lambda_{\bmb\bmc\bmd}$ and
  $\Lambda^*_{\bmb\bmc\bmd}$ are not independent. In fact,
  $\Lambda_{\bmb\bmc\bmd}=0$ if and only if $\Lambda^*_{\bmb\bmc\bmd}=0$.}
\end{remark}

\begin{remark}
{\em Equations \eqref{EvolutionElectricWeyl} and
  \eqref{EvolutionMagneticWeyl} imply a symmetric hyperbolic evolution
system for the (ten) independent components of the fields
$E_{\bma\bmb}$ and $B_{\bma\bmb}$ ---see e.g. \cite{AlhMenVal17} for
explicit expressions of the associated matrices.}
\end{remark}

\subsection{The subsidiary evolution system}
\label{Subsection:SubsidiarySystem}
The analysis of the relation between the solutions to the evolution
equations and actual solutions to the full conformal Einstein field
equations, the so-called \emph{propagation of the constraints},
requires the construction of a system of \emph{subsidiary evolution
equations for the zero-quantities} associated to the conformal
equations, \eqref{ecfe1}-\eqref{ecfe4}, and the gauge conditions
\eqref{Supplementary1}-\eqref{Supplementary3}. For the standard argument of the propagation of the
constraints to follow through, the subsidiary system is required to be
homogeneous in the zero-quantities. If this is the case, then it
follows from the uniqueness of solutions to symmetric hyperbolic
systems that if the zero-quantities vanish initially, then they will
vanish for all later times as the vanishing (zero) solution is always
a solution of a homogeneous evolution equation. 

\subsubsection{General remarks}
The basic assumption in the construction of the subsidiary evolution
system is that the evolution equations associated to the extended
conformal field equations are satisfied. Hence, we assume that
\[ 
\hat{\Sigma}{}_\bmzero{}^\bmc{}_\bmb=0, \qquad \hat{\Xi}{}^\bmc{}_{\bmd\bmzero\bmb}=0,
\qquad \hat{\Delta}{}_{\bmzero\bmb\bmc}=0,
\]
together with
\[
\Lambda{}_{(\bmi|\bmzero|\bmj)}=0, \qquad \Lambda{}^*{}_{(\bmi|\bmzero|\bmj)}=0.
\]
These evolution equations have been constructed using the gauge conditions
\[
f_\bmzero=0, \qquad \hat{\Gamma}{}_\bmzero{}^\bmb{}_\bmc=0, \qquad
\hat{L}{}_{\bmzero\bmb}=0. 
\]
These gauge conditions will also be used in the construction of the
subsidiary evolution system. Accordingly, the construction requires
the evolution equations for the additional zero-quantities
$\delta{}_a$, $\gamma{}_{ab}$ and $\varsigma{}_{ab}$ which are
associated to the gauge. In our gauge $d_0=0$ so that
\[
\delta_\bmzero=0.
\]
Since $\hat{L}{}_{\bmzero\bmb}=0$, by virtue of the definition of
$S{}_{ab}{}^{cd}$ and the evolution equation for the covector
$\beta_a$, namely, 
\[
\hat{\nabla}_\bmzero \beta_\bma + \beta_\bmzero \beta_\bma - \frac{1}{2}
\eta_{\bmzero\bma}(\beta_\bme \beta^\bme - 2 \lambda \Theta^{-2})=0, 
\]
it follows that
\[
\gamma{}_{\bmzero\bmb}= \hat{L}_{\bmzero\bmb} - \hat{\nabla}_\bmzero \beta_\bmb - \frac{1}{2}
S{}_{\bmzero\bmb}{}^{\bme\bmf} \beta_\bme \beta_\bmf + \lambda \Theta^{-2} \eta{}_{\bmzero\bmb}=0. 
\]
As a result of the $\Theta^{-2}$ in the last term of this equation, it
can only be used away from the conformal boundary ---this is, however,
not a problem in our analysis as the propagation of the constraints
only need to be considered in the regions where $\Theta\neq 0$. 
Moreover, by virtue of the gauge conditions and the evolution equation for the covector $f_a$, we have
\[
\varsigma{}_{\bmzero\bmb} = -\hat{L}{}_{\bmb\bmzero} - \hat{\nabla}{}_{\bmzero} f_\bmb +
\hat{\Gamma}{}_\bmb{}^\bme{}_\bmzero f_\bme =0. 
\]

\subsubsection{The subsidiary equation for the torsion}
To obtain the subsidiary equation for the no-torsion condition we
consider the totally antisymmetric covariant derivative $
\hat{\nabla}_{[\bma} \hat{\Sigma}{}_\bmb{}^\bmd{}_{\bmc]}$ and observe
that
\begin{equation}
3 \hat{\nabla}_{[\bmzero} \hat{\Sigma}{}_{\bmb}{}^\bmd{}_{\bmc]}=
\hat{\nabla}_\bmzero \hat{\Sigma}{}_\bmb{}^\bmd{}_\bmc -
\hat{\Gamma}{}_\bmb{}^\bme{}_\bmzero\hat{\Sigma}{}_\bmc{}^\bmd{}_\bme-
\hat{\Gamma}{}_\bmc{}^\bme{}_\bmzero\hat{\Sigma}{}_\bme{}^\bmd{}_\bmb. 
\label{ses1}
\end{equation}
On the other hand, from the first Bianchi identity
\[
\hat{R}{}^\bmd{}_{[\bmc\bma\bmb]} + \hat{\nabla}_{[\bma}\hat{\Sigma}{}_\bmb{}^\bmd{}_\bmc +
\hat{\Sigma}{}_{[\bma}{}^\bme{}_\bmb \hat{\Sigma}{}_{\bmc]}{}^\bmd{}_\bme=0, 
\]
and the definition of $\hat{\Xi}{}^\bmd{}_{\bmc\bma\bmb}$ one obtains
\begin{equation}
\hat{\nabla}_{[\bma}
  \hat{\Sigma}{}_\bmb{}^\bmd{}_{\bmc]}=-\hat{\Xi}{}^\bmd{}_{[\bmc\bma\bmb]}-
  \hat{\Sigma}{}_{[\bma}{}^\bme{}_\bmb \hat{\Sigma}{}_{\bmc]}{}^\bmd{}_\bme, 
\label{ses2} 
\end{equation}
where it has been used that, by construction,
$\hat{\rho}{}^\bmd{}_{[\bmc\bma\bmb]}=0$. The desired evolution
equation is obtained combining equations \eqref{ses1} and
\eqref{ses2} to yield
\begin{equation} 
 \hat{\nabla}_{\bmzero} \hat{\Sigma}{}_{\bmb}{}^\bmd{}_{\bmc}=-
\frac{1}{3}\hat{\Gamma}{}_\bmc{}^\bme{}_\bmzero \hat{\Sigma}{}_\bme{}^\bmd{}_\bmb -
\frac{1}{3}\hat{\Gamma}{}_\bmc{}^\bme{}_\bmzero\hat{\Sigma}{}_\bme{}^\bmd{}_\bmb -
\hat{\Xi}{}^\bmd{}_{\bmzero\bmb\bmc}. 
\label{ses3}
\end{equation}
This evolution equation is homogeneous in the various
zero-quantities. 

\subsubsection{The subsidiary equation for the Ricci identity}
To obtain a subsidiary equation for the Ricci identity, we consider
the totally symmetrised covariant derivative
$\hat{\nabla}{}_{[\bma}\hat{\Xi}{}^\bmd{}_{|\bme|\bmb\bmc]}$ and observe that
\begin{equation} 
3\hat{\nabla}{}_{[\bmzero}\hat{\Xi}{}^\bmd{}_{|\bme|\bmb\bmc]}=\hat{\nabla}_\bmzero\hat{\Xi}{}^\bmd{}_{\bme\bmb\bmc}-\hat{\Gamma}{}_\bmb{}^\bmf{}_\bmzero\hat{\Xi}{}^\bmd{}_{\bme\bmc\bmf}-\hat{\Gamma}{}_\bmc{}^\bmf{}_\bmzero\hat{\Xi}{}^\bmd{}_{\bme\bmf\bmb}.
\label{ses4}
\end{equation}
Using the second Bianchi identity
\[
\hat{\nabla}_{[\bma}\hat{R}{}^\bmd{}_{|\bme|\bmb\bmc]} + \hat{\Sigma}{}_{[\bma}{}^\bmf{}_\bmb \hat{R}{}^\bmd{}_{|\bme|\bmc]\bmf}=0 
\]
and the definition of $\hat{\Xi}{}^\bmd{}_{\bme\bmb\bmc}$ it follows that
\begin{equation} 
\hat{\nabla}{}_{[\bma}\hat{\Xi}{}^\bmd{}_{|\bme|\bmb\bmc]}= -
\hat{\Sigma}{}_{[\bma}{}^\bmf{}_\bmb \hat{R}{}^\bmd{}_{|\bme|\bmc]\bmf}-
\hat{\nabla}_{[\bma}\hat{\rho}{}^\bmd{}_{|\bme|\bmb\bmc]}. 
\label{ses5}
\end{equation}
The first term on the right-hand side is already of the required
form. The second one needs to be analysed in more detail. For this, it is recalled that
\[
\hat{\rho}{}^\bmd{}_{\bme\bmb\bmc} \equiv C{}^\bmd{}_{\bme\bmb\bmc}+ 2
S{}_{\bme[\bmb}{}^{\bmd\bmf}\hat{L}_{\bmc]\bmf}. 
\]
Thus,
\[
\hat{\nabla}{}_{[\bma}\hat{\rho}{}^\bmd{}_{|\bme|\bmb\bmc]}=\hat{\nabla}_{[\bma}C{}^\bmd_{|\bme|\bmb\bmc]}+
2S{}_{\bme[\bmb}{}^{\bmd\bmf}\hat{\nabla}{}_\bma \hat{L}{}_{\bmc]\bmf}. 
\]
To further expand this expression we consider the combination $\epsilon{}_\bmf{}^{\bma\bmb\bmc}\hat{\nabla}_\bma\hat{\rho}{}^\bmd{}_{\bme\bmb\bmc}$. A direct computation shows that
\[
\hat{\nabla}{}_{[\bma}C{}^\bmd{}_{|\bme|\bmb\bmc]}=\nabla{}_{[\bma}C{}^\bmd{}_{|\bme|\bmb\bmc]}+
\delta{}_{[\bma}{}^\bmd f{}_{|\bmf}C{}^\bmf{}_{\bme|\bmb\bmc]}+\eta{}_{\bme[\bma}f^\bmf
C{}^\bmd{}_{|\bmf|\bmb\bmc]}. 
\]
Moreover, one has
\[
\epsilon{}_\bmf{}^{\bma\bmb\bmc}\nabla_\bma C{}^\bmd{}_{\bme\bmb\bmc}= -
\epsilon{}_\bme{}^{\bmd\bmg\bmh}\nabla_\bma C{}^\bma{}_{\bmf\bmg\bmh}. 
\]
Thus, using that $C{}^\bmc{}_{\bmd\bma\bmb}= \Theta d{}^\bmc{}_{\bmd\bma\bmb}$ and the definition of the zero quantity $\Lambda_{\bma\bmb\bmc}$ it follows that
\[
\epsilon{}_\bmf{}^{\bma\bmb\bmc}\hat{\nabla}{}_\bma C{}^\bmd{}_{\bme\bmb\bmc}=\Theta
\epsilon{}_\bme{}^{\bmd\bmg\bmh}\Lambda{}_{\bmf\bmg\bmh} + 2 \nabla^\bmg \Theta
d{}^{*\bmd}{}_{\bme\bmf\bmg}+ 2 \Theta f^\bmg d{}^*{}_{\bmg\bme\bmf}{}^\bmd + 2 \Theta f^\bmg
d^{*\bmd}{}_{\bmg\bmf\bme}. 
\]
A similar computation using the definition of $\hat{\Delta}{}_{\bma\bmb\bmc}$ yields
\[
 2 \epsilon{}_\bmf{}^{\bma\bmb\bmc}S{}_{\bme\bmb}{}^{\bmd\bmg}\hat{\Delta}_{\bma\bmc\bmg}= 2 \Theta
 \beta_\bmg d^{*\bmg}{}_{\bme\bmf}{}^\bmd - 2 \Theta \beta_{\bmg} d^{*\bmg\bmd}{}_{\bmf\bme}. 
\]
Thus, using the symmetries of $d^{*}{}_{\bmc\bmd\bma\bmb}$ and the
definition of $\delta_\bma$ one concludes that
\[
\epsilon{}_\bmf{}^{\bma\bmb\bmc}\hat{\nabla}_\bma \hat{\rho}{}^\bmd{}_{\bme\bmb\bmc}=\Theta
\epsilon{}_\bme{}^{\bmd\bmg\bmh}\Lambda_{\bmf\bmg\bmh} - 2 \Theta \delta^\bmg
d{}^{*\bmd}{}_{\bme\bmf\bmg}+ \epsilon{}_\bmf{}^{\bma\bmb\bmc}
S{}_{\bme\bmb}{}^{\bmd\bmg}\hat{\Delta}_{\bma\bmc\bmg}. 
\]
Alternatively, using the properties of the generalised Hodge duals we can write
\[
\hat{\nabla}_{[\bma} \hat{\rho}{}^\bmd{}_{|\bme|\bmb\bmc]}= \frac{1}{6} \Theta
\epsilon{}^\bmf{}_{\bma\bmb\bmc} \epsilon{}_\bme{}^{\bmd\bmg\bmh}\Lambda_{\bmf\bmg\bmh} - \frac{1}{3}
\Theta \epsilon{}^\bmf{}_{\bma\bmb\bmc} \delta^\bmg d{}^{*\bmd}{}_{\bme\bmf\bmg}-
S{}_{\bme[\bmb}{}^{\bmd\bmg}\hat{\Delta}_{\bma\bmc]\bmg}. 
\]
Combining the expressions, we obtain the following evolution equation
\begin{equation} 
\begin{split} 
\hat{\nabla}_\bmzero \hat{\Xi}{}^\bmd{}_{\bme\bmb\bmc}=&\hat{\Gamma}{}_\bmb{}^\bmf{}_\bmzero \hat{\Xi}{}^\bmd{}_{\bme\bmc\bmf} + \hat{\Gamma}{}_\bmc{}^\bmf{}_\bmzero \hat{\Xi}{}^\bmd{}_{\bme\bmf\bmb}- \hat{\Sigma}{}_\bmb{}^\bmf{}_\bmc \hat{R}{}^\bmd{}_{\bme\bmzero\bmf} - \frac{1}{2} \Theta \epsilon{}^\bmf{}_{\bmzero\bmb\bmc} \epsilon{}_\bme{}^{\bmd\bmg\bmh}\Lambda{}_{\bmf\bmg\bmh} \\
&+ \epsilon{}^\bmf{}_{\bmzero\bmb\bmc} \delta^\bmg d{}^{*\bmd}{}_{\bme\bmf\bmg} + 3S{}_{\bme\bmzero}{}^{\bmd\bmg}
\hat{\Delta}{}_{\bmc\bmb\bmg}, \label{ses6}  
\end{split}
\end{equation}
which is homogeneous in the zero-quantities.

\subsubsection{Subsidiary equation for the Cotton equation}
Now, to compute the subsidiary equation for the Cotton equation we
consider $\hat{\nabla}{}_{[\bma} \hat{\Delta}{}_{\bmb\bmc]\bmd}$. On
the one hand, a direct computation yields
\[
3 \hat{\nabla}{}_{[\bmzero} \hat{\Delta}{}_{\bmb\bmc]\bmd}= \hat{\nabla}{}_{\bmzero}
\hat{\Delta}{}_{\bmb\bmc\bmd}- \hat{\Gamma}{}_\bmb{}^\bme{}_\bmzero \hat{\Delta}_{\bmc\bme\bmd} -
\hat{\Gamma}{}_\bmc{}^\bme{}_\bmzero \hat{\Delta}_{\bme\bmb\bmd}. 
\]
On the other hand, using the definition of
$\hat{\Xi}{}^\bme{}_{\bmc\bma\bmb}$ and the symmetries of
$\hat{\rho}^\bme{}_{\bmc\bma\bmb}$ one obtains
\[ 
 \hat{\nabla}{}_{[\bma} \hat{\Delta}{}_{\bmb\bmc]\bmd}=-
 \hat{\Xi}{}^\bme{}_{[\bmc\bma\bmb]}\hat{L}{}_{\bme\bmd} -
 \hat{\Xi}{}^\bme{}_{\bmd[\bma\bmb}\hat{L}_{\bmc]\bme} - \hat{\rho}{}^\bme{}_{\bmd[\bma\bmb}
 \hat{L}{}_{\bmc]\bme} + \hat{\Sigma}{}_{[\bma}{}^\bme{}_\bmb \hat{\nabla}{}_{|\bme|}
 \hat{L}_{\bmc]\bmd} - \hat{\nabla}{}_{[\bma} d_{|\bme} d{}^\bme{}_{\bmd|\bmb\bmc]}- d_\bme
 \hat{\nabla}{}_{[\bma} d{}^\bme{}_{|\bmd|\bmb\bmc]}. 
\]
Using the definition of $\delta_a$ and $\gamma{}_{ab}$ one finds that
\[
\hat{\nabla}_{[\bma}d_{|\bme} d{}^\bme{}_{\bmd|\bmb\bmc]}= - \Theta
\delta_{[\bma}\beta_{|\bme}d{}^\bme{}_{\bmd|\bmb\bmc]} - \Theta
\gamma_{[\bma|\bme}d{}^\bme{}_{\bmd|\bmb\bmc]} - \Theta f_{[\bma}
\beta_{|\bme}d{}^\bme{}_{\bmd|\bmb\bmc]}+ \Theta \hat{L}_{[\bma|\bme} d{}^\bme{}_{\bmd|\bmb\bmc]}. 
\]
Finally, a calculation shows that $\epsilon{}_{\bmf}{}^{\bma\bmb\bmc} \nabla_\bma d{}^\bme{}_{\bmd\bmb\bmc} = \epsilon{}_{\bmd}{}^{\bme\bmg\bmh} \nabla_\bma d{}^\bme{}_{\bmf\bmg\bmh}$, 
so that using 
\[
\hat{\nabla}_{[\bma} C{}^\bmd{}_{|\bme|\bmb\bmc]}= \nabla_{[\bma} C{}^\bmd{}_{|\bme|\bmb\bmc]} +
\delta{}_{[\bma}{}^\bmd f_{|\bmf} C{}^\bmf{}_{\bme|\bmb\bmc]} + \eta_{\bme[\bma} f^\bmf
C{}^\bmd{}_{|\bmf|\bmb\bmc]}, 
\]
and the properties of the generalised duals we find that
\[
\hat{\nabla}_{[\bma} d{}^\bme{}_{|\bmd|\bmb\bmc]} = \frac{1}{6} \epsilon{}_{\bma\bmb\bmc}{}^\bmf
\epsilon{}_\bmd{}^{\bme\bmg\bmh}\Lambda{}_{\bmf\bmg\bmh} + \delta{}_{[\bma}{}^\bme f_{|\bmf}
d{}^\bmf{}_{\bmd|\bmb\bmc]}+ \eta{}_{\bmd[\bma} f^\bmf d{}^\bme{}_{|\bmf|\bmb\bmc]}. 
\]
Combining the above expressions and using the properties of the decomposition of $\hat{\rho}{}^\bmc{}_{\bmd\bma\bmb}$ we obtain the expression
\[
\hat{\nabla}_{[\bma} \hat{\Delta}{}_{\bmb\bmc]\bmd}= -
\hat{\Xi}{}^\bme{}_{[\bmc\bma\bmb]}\hat{L}_{\bme\bmd}- \hat{\Xi}{}^\bme{}_{\bmd[\bma\bmb}
\hat{L}{}_{\bmc]\bme} + \hat{\Sigma}{}_{[\bma}{}^{\bme}{}_\bmb \hat{\nabla}_{|\bme|}
\hat{L}{}_{\bmc]\bmd} + \Theta \delta_{[\bma} \beta_{|\bme} d{}^\bme{}_{\bmd|\bmb\bmc]} +
\Theta \gamma{}_{[\bma|\bme} d{}^\bme{}_{\bmd|\bmb\bmc]} - \frac{1}{6}
\epsilon{}_{\bma\bmb\bmc}{}^\bmf \epsilon{}_\bmd{}^{\bme\bmg\bmh}\Lambda{}_{\bmf\bmg\bmh} \beta_\bme 
\]
and, eventually, the evolution equation
\[
\begin{split}
\hat{\nabla}_\bmzero \hat{\Delta}_{\bmb\bmc\bmd}=& \hat{\Gamma}{}_\bmb{}^\bme{}_\bmzero
\hat{\Delta}{}_{\bmc\bme\bmd}+ \hat{\Gamma}{}_\bmc{}^\bme{}_\bmzero \hat{\Delta}{}_{\bme\bmb\bmd}-
\hat{\Xi}{}^\bme{}_{\bmzero\bmb\bmc} \hat{L}{}_{\bme\bmd} + \delta_\bmb d_\bme d{}^\bme{}_{\bmd\bmc\bmzero} +
\delta_\bmc d_\bme d{}^\bme{}_{\bmd\bmzero\bmb} \\
&+ \Theta \gamma{}_{\bmb\bme} d{}^\bme{}_{\bmd\bmc\bmzero} +
\Theta \gamma{}_{\bmc\bme} d{}^\bme{}_{\bmd\bmzero\bmb} - \frac{1}{2}\epsilon{}_{\bmzero\bmb\bmc}{}^\bmf
\epsilon{}_\bmd{}^{\bme\bmg\bmh} \Lambda{}_{\bmf\bmg\bmh} \beta_\bme, 
\end{split}
\]
which is homogeneous in zero-quantities as required.

\subsubsection{Subsidiary equations for the Bianchi identity}
Finally, we are left to show the propagation of the physical Bianchi
identity. In view of the contracted derivative appearing in this
equation, the construction of suitable subsidiary equations is more
involved. 

\medskip
Since $h{}_{a}{}^b= \delta{}_a{}^b + \tau_a \tau^b$, it follows then
that 
\begin{equation} 
\Lambda{}_{abc}= \delta{}_a{}^d
\Lambda{}_{dbc}=(h{}_{a}{}^d- \tau{}_a
\tau{}^d)\Lambda{}_{dbc}=h{}_{a}{}^d\Lambda{}_{dbc}- \tau{}_a
\tau{}^d\Lambda{}_{dbc}. 
\label{ses7}
\end{equation}
Now, let 
\[
\Omega{}_{abc}\equiv h{}_{a}{}^d\Lambda{}_{dbc}, \qquad
\Omega{}_{bc}\equiv \tau{}^d\Lambda{}_{dbc}.
\]
By construction, the tensor $\Omega{}_{bc}$ is antisymmetric, hence it
admits a decomposition in \emph{electric} and \emph{magnetic
  parts}. That is, one can write 
\[
\Omega{}_{bc}= \Omega{}_{[bc]}= \Omega^*_e \epsilon{}^e{}_{bc}- 2
\Omega{}_{[b}\tau{}_{c]}, 
\]
where
\[
\Omega_a \equiv \Omega{}_{cb}\tau^b h{}_a{}^c, \qquad \Omega^*_a \equiv \Omega^*_{cb}\tau^b h{}_a{}^c.
\]
Furthermore, one also has that 
\[
\Omega{}_{dbc}= \Omega{}_{d[bc]}= H^*_{de} \epsilon{}^e{}_{dc}-
2H{}_{d[b} \tau{}_{c]},
\]
where
\[
H{}_{da}\equiv \Omega{}_{dcb}\tau^b h{}_{a}{}^c,\qquad H{}^*{}_{da}\equiv \Omega^*_{dcb}\tau^b h{}_{a}{}^c.
\]
Substituting the above expressions for $\Omega{}_{bc}$ and
$\Omega{}_{dbc}$ into equation \eqref{ses7} it follows then that
\begin{equation}
\Lambda{}_{abc}=h{}_a{}^d(H^*_{de} \epsilon{}^e{}_{dc}- 2H{}_{d[b}
n{}_{c]})- n{}_a( \Omega{}^*_e \epsilon{}^e{}_{bc}- 2
\Omega{}_{[b}n{}_{c]}). 
\label{ses8}
\end{equation}
Crucially, it can be verified that if the evolution equations \eqref{EvolutionElectricWeyl}
and \eqref{EvolutionMagneticWeyl} for the electric and magnetic part sof the rescaled Weyl
tensor are satisfied then 
\[
H{}_{da}=0, \qquad H{}^*{}_{da}=0.
\]
If the above holds, then equation \eqref{ses8} reduces to 
\[
\Lambda{}_{abc}= n{}_a(2 \Omega{}_{[b}n{}_{c]}- \Omega{}^*{}_e
\epsilon{}^e{}_{bc})= n{}_a \Omega{}_{bc}. 
\]

\begin{remark}
{\em The tensors $\Omega_a$ and $\Omega^*_a$ encode, respectively, the \emph{Gauss
  constraints} for the electric and magnetic parts of the Weyl tensor
---that is, the equations
\[
\mathcal{D}^a d_{ab}=0, \qquad \mathcal{D}^ad_{ab}^*=0.
\]}
\end{remark}

To conclude the computation, it remains to compute $\nabla{}^a
\Lambda{}_{abc}$.  A direct calculation gives
\begin{equation}
\nabla{}^a \Lambda{}_{abc}=  \nabla{}^a \tau{}_a
\Omega{}_{bc}+\tau{}_a \nabla{}^a\Omega{}_{bc}=\nabla{}^a \tau{}_a
\Omega{}_{bc}+\partial{}_\tau\Omega{}_{bc}. 
\label{ses9}
\end{equation}
An alternative computation of $\nabla{}^a \Lambda{}_{abc}$ using the commutator of the covariant derivative $\nabla$ gives
\[
2 \nabla^b \Lambda{}_{bcd}= 2 \nabla{}^{[b} \nabla{}^{a]}d{}_{abcd} =
2 R{}^e{}_{[c}{}^{ba}d{}_{d]eab} - 2 R{}^e{}_a{}^{ba}d{}_{ebcd} +
\Sigma{}_b{}^e{}_a \nabla_e d{}^{ab}{}_{cd}. 
\]
Observing that $\hat{\Sigma}{}_a{}^c{}_b = \Sigma{}_a{}^c{}_b$ as
$\hat{\nabla} - \nabla= S(f)$, it follows that the equation 
\[
\hat{R}{}^a{}_{bcd} - {R}{}^a{}_{bcd}= 2(\delta{}^a{}_{[c}
\hat{\nabla}{}_{d]}\hat{f}_b + \hat{\nabla}{}_{[c} \hat{f}{}^a
\hat{g}{}_{d]b} - \delta{}^a{}_b \hat{\nabla}{}_{[c} \hat{f}{}_{d]} -
\delta{}^a{}_{[c} \hat{f}{}_{d]} \hat{f}_b +
\hat{g}{}_{b[c}\hat{f}{}_{d]}\hat{f}^a + \delta{}^a{}_{[c}
\hat{g}{}_{d]b} \hat{f}_e \hat{f}^e ) 
\]
together with the definitions of the zero quantities
$\hat{\Xi}{}^\bmc{}_{\bmd\bma\bmb}$ and $\varsigma{}_{\bma\bmb}$ and
the symmetries of $d{}_{\bma\bmb\bmc\bmd}$ so that after projecting
the equations with respect to the frame one obtains
\begin{equation}
\nabla{}^\bmb \Lambda{}_{\bmb\bmc\bmd} = \hat{\Xi}{}^\bme{}_{[\bmc}{}^{\bmb\bma} d{}_{\bmd]\bme\bma\bmb} -
\hat{\Xi}{}^\bme{}_\bma{}^{\bmb\bma} d{}_{\bme\bmb\bmc\bmd} + \frac{1}{2}
\hat{\Sigma}{}_\bmb{}^\bme{}_\bma \nabla_\bme d{}^{\bma\bmb}{}_{\bmc\bmd} +
\varsigma{}^{\bma\bmb}d{}_{\bma\bmb\bmc\bmd}, \label{ses10} 
\end{equation}
which is homogeneous in zero-quantities. Hence, combining equations
\eqref{ses9} and \eqref{ses10}, we obtain the following equation for
the components of $\Omega_{ab}$:
\[
\partial{}_\bmzero\Omega{}_{\bmb\bmc} = \hat{\Xi}{}^\bme{}_{[\bmb}{}^{\bma\bmf}
d{}_{\bmc]\bme\bmf\bma} - \hat{\Xi}{}^\bme{}_\bmf{}^{\bma\bmf} d{}_{\bme\bma\bmb\bmc} + \frac{1}{2}
\hat{\Sigma}{}_\bma{}^\bme{}_\bmf \nabla_\bme d{}^{\bmf\bma}{}_{\bmb\bmc} +
\varsigma{}^{\bmf\bma}d{}_{\bmf\bma\bmb\bmc}-\chi \Omega{}_{\bmb\bmc}.
\]

\subsubsection{Subsidiary equations for the gauge conditions}
To conclude our discussion of the subsidiary equations, we are left
with the task of providing evolution equations for the zero-quantities
associated to the gauge. In order to do so we expand
$\hat{\nabla}{}_{[\bmzero} \delta{}_{\bmb]}$,
$\hat{\nabla}{}_{[\bmzero} \gamma{}_{\bmb]\bmc}$ and
$\hat{\nabla}{}_{[\bmzero} \varsigma{}_{\bmb\bmc]}$ to get
\begin{eqnarray*}
&& 2\hat{\nabla}{}_{[\bmzero} \delta{}_{\bmb]}= \hat{\nabla}{}_\bmzero \delta_\bmb + \hat{\Gamma}{}_\bmb{}^\bme{}_\bme \delta_\bme, \\
&& 2\hat{\nabla}{}_{[\bmzero} \gamma{}_{\bmb]\bmc}= \hat{\nabla}{}_\bmzero \gamma_{\bmb\bmc} + \hat{\Gamma}{}_\bmb{}^\bme{}_\bmzero \gamma_{\bme\bmc}, \\
&& 3\hat{\nabla}{}_{[\bmzero} \varsigma{}_{\bmb\bmc]}= \hat{\nabla}{}_\bmzero \varsigma_{\bmb\bmc} - \hat{\Gamma}{}_\bmb{}^\bme{}_\bmzero \varsigma_{\bmc\bme} - \hat{\Gamma}{}_\bmc{}^\bme{}_\bmzero \varsigma_{\bme\bmb}.
\end{eqnarray*}
We then compute $\hat{\nabla}{}_{[a}\delta{}_{b]}$,
$\hat{\nabla}{}_{[a}\gamma{}_{b]c}$ and
$\hat{\nabla}{}_{[a}\varsigma{}_{bc]}$ explicitly making use of the
definitions of the zero-quantities and re-expressing the result in
terms of zero-quantities so as to obtain
\begin{eqnarray*}
&& 2 \hat{\nabla}{}_{[\bma}\delta{}_{\bmb]}= - \gamma{}_{[\bma\bmb]}+
   \varsigma_{\bma\bmb} - \frac{1}{2} \Theta^{-1} \Sigma{}_\bma{}^\bme{}_\bmb
   \hat{\nabla}_\bme \Theta, \\
&& 2 \hat{\nabla}_{[\bma} \gamma{}_{\bmb]\bmc}= \hat{\Delta}{}_{\bma\bmb\bmc} + \beta_\bme \hat{\Xi}{}^\bme{}_{\bmc\bma\bmb} - \hat{\Sigma}{}_\bma{}^\bme{}_\bmb \hat{\nabla}_\bme \beta_\bmc + 2 \beta_\bmc \gamma_{[\bma\bmb]} - 2 \beta_{[\bma} \gamma{}_{\bmb]\bmc}\\
&& \hspace{2cm}+ \eta{}_{\bmc[\bma} \beta^\bme \gamma_{\bmb]\bme} + 2 \lambda
   \Theta^{-2} \delta_{[\bma} \eta{}_{\bmb]\bmc} + \beta_{[\bma}\eta_{\bmb]\bmc} \beta_\bme
   \beta^\bme - 2 \lambda \Theta^{-2} \eta_{\bmc[\bma} \beta_{\bmb]}, \\
&& \hat{\nabla}_{[\bma}\varsigma{}_{\bmb\bmc]}= \frac{1}{2} \hat{\Delta}_{[\bma\bmb\bmc]} + \frac{1}{2} \hat{\Xi}{}^\bme{}_{[\bmc\bma\bmb]}f_\bme - \frac{1}{2} \hat{\Sigma}{}_{[\bma}{}^\bme{}_\bmb \hat{\nabla}_{|\bme|} f_{\bmc]}.
\end{eqnarray*}
From the above expressions it follows that the evolution equations for
$\delta_a$, $\gamma_{ab}$ and $\varsigma_{ab}$ are given by 
\begin{subequations}
\begin{eqnarray}
 &&\hat{\nabla}_\bmzero \delta_\bmi= \gamma_{\bmi\bmzero}- \hat{\Gamma}{}_\bmi{}^\bme{}_\bmzero \delta_\bme , \label{SubsidiaryGauge1}\\
&& \hat{\nabla}_\bmzero \gamma_{\bmi\bmc}= - \gamma_{\bmj\bmc} \hat{\Gamma}{}_\bmi{}^\bmj{}_\bmzero
   - \beta_\bmzero \gamma_{\bmi\bmc} - \beta_\bmc \gamma_{\bmi\bmzero} + \eta_{\bmzero\bmc}(\beta^\bme
   \gamma_{\bmi\bme} - 2 \lambda \Theta^{-2} \delta_\bmi), \label{SubsidiaryGauge2}\\
&& \hat{\nabla}_\bmzero \varsigma_{\bmj\bmk} =
   \hat{\Gamma}{}_\bmj{}^\bme{}_\bmzero \varsigma{}_{\bmk\bme} +
   \Gamma{}_\bmk{}^\bme{}_\bmzero \varsigma{}_{\bme\bmj} + \frac{1}{2}
   \hat{\Delta}_{\bmj\bmk\bmzero} + \frac{1}{2}
   \hat{\Xi}{}^\bme{}_{\bmzero\bmj\bmk} f_\bme + \frac{1}{2}
   \hat{\Sigma}{}_\bmj{}^\bme{}_\bmk
   \hat{\Gamma}{}_\bme{}^\bmf{}_\bmzero f_\bmf, \label{SubsidiaryGauge3}
\end{eqnarray}
\end{subequations}
where, in particular, the evolution equation for the covector $\beta_a$,
\[
 \hat{\nabla}_\bmzero \beta_\bma + \beta_\bmzero \beta_\bma -
  \frac{1}{2} \eta_{\bmzero\bma}(\beta_\bme \beta^\bme - 2 \lambda
  \Theta^{-2})=0,
\]
has been used in the derivation of
equation \eqref{SubsidiaryGauge2}. Again, as required, the equations
\eqref{SubsidiaryGauge1}-\eqref{SubsidiaryGauge3} are homogeneous in
various zero-quantities. 

\begin{remark}
{\em Observe that equation \eqref{SubsidiaryGauge2} contains the
  potentially singular term  $\lambda \Theta^{-2} \delta_\bmi$. As
  such, this equation can only be used away from the conformal
  boundary where $\Theta\neq 0$. This is a consequence of the use of a
conformal Gaussian gauge hinged on a standard Cauchy
hypersurface. This singular behaviour is of no consequence in our
analysis as one is only interested on solutions to the subsidiary
equations away from the conformal boundary.}
\end{remark}

\subsection{Summary: structural properties of the evolution and subsidiary equations}
As conclusion of the long computations in this section, we now provide
a summary of the conformal evolution equations, the associated
subsidiary system and the structural properties of these systems which
will be required in the reminder of our analysis.

\medskip
The computations discussed in the previous subsections show that in a
conformal Gaussian gauge the various fields associated to the extended
vacuum conformal Einstein field equations satisfy the evolution equations
\begin{subequations}
\begin{eqnarray}
&& \partial_\tau e_\bmb{}^\nu = - \hat{\Gamma}{}_\bmb{}^\bmc{}_\bmzero e_\bmc{}^\nu, \label{XCFEEvolution1}\\
&&\partial_\tau \hat{L}{}_{\bmd\bmb}=\hat{\Gamma}{}_\bmzero{}^\bmc{}_\bmd \hat{L}{}_{\bmc\bmb}+ \hat{\Gamma}{}_\bmzero{}^\bmc{}_\bmb \hat{L}{}_{\bmd\bmc} + d_\bma \hat{d}{}^\bma{}_{\bmb\bmzero\bmd}, \label{XCFEEvolution2}\\
&&  \partial_\tau f_\bmi = - f_\bmj \hat{\Gamma}{}_\bmi{}^\bmj{}_\bmzero + \hat{L}{}_{\bmi\bmzero}, \label{XCFEEvolution3}\\
&& \partial_\tau (\hat{\Gamma}{}_\bmb{}^\bmc{}_\bmd) = -\hat{\Gamma}{}_\bmf{}^\bmc{}_\bmd \hat{\Gamma}{}_\bmb{}^\bmf{}_\bmzero-\Xi \hat{d}{}^\bmc{}_{\bmd\bmzero\bmb}-2{\delta_\bmd}^\bmc \hat{L}_{\bmb\bmzero}-2{\delta_\bmzero}^c \hat{L}_{\bmb\bmd}+2 g_{\bmd\bmzero}g^{\bmc\bme}\hat{L}_{\bmb\bme},\label{XCFEEvolution4}\\
&& \partial_\tau d_{\bmb\bmd}+\epsilon^{\bme\bmf}{}_{(\bmd}D_\bmf d^*_{\bmb)\bme}= 2a_\bmf\epsilon^{\bme\bmf}_{(\bmd} d^*_{\bmb)\bme}- \chi d_{\bmb\bmd} +2 \chi^\bmf{}_{(\bmb} d_{\bmd)\bmf}, \label{XCFEEvolution5}\\
&&  \partial_\tau d^*_{\bmb\bmd} -\epsilon^\bme{}_{\bmf(\bmd}D^\bmf
   d_{\bmb)\bme} = 2 a^\bmf {\epsilon_{\bmf(\bmd}}^\bme d_{\bmb)\bme} - \chi
   d^*_{\bmb\bmd} + 2 \chi^\bmf_{(\bmb} d^*_{\bmd)\bmf}. \label{XCFEEvolution6}
\end{eqnarray}
\end{subequations}
Letting $\bme$, $\bmGamma$, $\hat{\bmL}$ and $\bmphi$ denote,
respectively, the independent components of the coefficients of the
frame, the connection coefficients, the Schouten tensor of the Weyl
connection and the rescaled Weyl tensor and setting, for convenience,
$\hat{\mathbf{u}}\equiv ( \hat{\bmupsilon},\hat{\bmphi})$,
$\hat{\bmupsilon}\equiv  (\bme, \hat{\bmGamma}, \hat{\bmL})$, one has
that equations \eqref{XCFEEvolution1}-\eqref{XCFEEvolution6} can be written, schematically, in the form
\begin{subequations}
\begin{eqnarray}
&& \partial_\tau \hat{\bmupsilon}= \mathbf{K} \hat{\bmupsilon} + \mathbf{Q}(\hat{\bmupsilon}, \hat{\bmupsilon}) + \mathbf{L}(\bar{x})
   \hat{\bmphi}, \label{he1}  \\
&& (\mathbf{I} + \mathbf{A}^0(\bme)) \partial_\tau \hat{\bmphi} + \mathbf{A}^\alpha (\bme) \partial_\alpha \hat{\bmphi}
   = \mathbf{B}(\hat{\bmGamma}) \hat{\bmphi}, \label{he2}
\end{eqnarray}
\end{subequations}
where $\mathbf{K}$ and $\mathbf{Q}$ denote, respectively, a matrix and
a quadratic form, both with constant coefficients while $\mathbf{L}$
is a matrix with coefficients depending smoothly on the
coordinates. Moreover, $\mathbf{A}^\mu(\bme)$ denote, for
$\mu=0,\ldots,3$ Hermitian matrix-valued functions depending smoothly on $\bme$. In
particular  $\mathbf{I} + \mathbf{A}^0(\bme)$ is positive definite for
 $\bme$ suitably close to the background solution ---with closeness
 understood in the sense of Sobolev norms. Finally, $\mathbf{B}(\hat{\bmGamma})$ denotes a
smooth matrix-value function of the component of the connection. 

\begin{remark}
{\em Altogether, the conformal evolution system described by equations
\eqref{he1}-\eqref{he2} constitutes a quasilinear symmetric hyperbolic
system for which a well-posedness theory is available ---see
\cite{Kat75}, also \cite{CFEBook} for an abridged version. This theory
will be used in the remaining sections of this article to establish
the stability of the solution to the Einstein field equations given by
the metric \eqref{BackgroundPhysicalMetric}.}
\end{remark}

\begin{remark}
{\em A remarkable structural property of the conformal evolution
  system \eqref{he1}-\eqref{he2} is that the equations in \eqref{he1}
  are, in fact, mere transport equations along conformal
  geodesics. The true hyperbolic content of the system is contained in
the \emph{Bianchi subsystem} \eqref{he2}. This property does not play
any particular role in our analysis, but it may prove key in, for example, the
analysis of formation of singularities.}
\end{remark}

Regarding the subsidiary evolution system, the key conclusion from
the system
\begin{subequations}
\begin{eqnarray}
&& \hat{\nabla}_{\bmzero} \hat{\Sigma}{}_{\bmb}{}^\bmd{}_{\bmc}=-
\frac{1}{3}\hat{\Gamma}{}_\bmc{}^\bme{}_\bmzero \hat{\Sigma}{}_\bme{}^\bmd{}_\bmb -
\frac{1}{3}\hat{\Gamma}{}_\bmc{}^\bme{}_\bmzero\hat{\Sigma}{}_\bme{}^\bmd{}_\bmb -
\hat{\Xi}{}^\bmd{}_{\bmzero\bmb\bmc}, \label{SubsidaryEqnFirst}\\
&& \hat{\nabla}_\bmzero \hat{\Xi}{}^\bmd{}_{\bme\bmb\bmc}=\hat{\Gamma}{}_\bmb{}^\bmf{}_\bmzero \hat{\Xi}{}^\bmd{}_{\bme\bmc\bmf} + \hat{\Gamma}{}_\bmc{}^\bmf{}_\bmzero \hat{\Xi}{}^\bmd{}_{\bme\bmf\bmb}- \hat{\Sigma}{}_\bmb{}^\bmf{}_\bmc \hat{R}{}^\bmd{}_{\bme\bmzero\bmf} - \frac{1}{2} \Theta \epsilon{}^\bmf{}_{\bmzero\bmb\bmc} \epsilon{}_\bme{}^{\bmd\bmg\bmh}\Lambda{}_{\bmf\bmg\bmh} \\
&& \hspace{2cm} + \epsilon{}^\bmf{}_{\bmzero\bmb\bmc} \delta^\bmg d{}^{*\bmd}{}_{\bme\bmf\bmg} + 3S{}_{\bme\bmzero}{}^{\bmd\bmg}
\hat{\Delta}{}_{\bmc\bmb\bmg}, \\
&& \hat{\nabla}_\bmzero \hat{\Delta}_{\bmb\bmc\bmd}= \hat{\Gamma}{}_\bmb{}^\bme{}_\bmzero
\hat{\Delta}{}_{\bmc\bme\bmd}+ \hat{\Gamma}{}_\bmc{}^\bme{}_\bmzero \hat{\Delta}{}_{\bme\bmb\bmd}-
\hat{\Xi}{}^\bme{}_{\bmzero\bmb\bmc} \hat{L}{}_{\bme\bmd} + \delta_\bmb d_\bme d{}^\bme{}_{\bmd\bmc\bmzero} +
\delta_\bmc d_\bme d{}^\bme{}_{\bmd\bmzero\bmb} \\
&&  \hspace{2cm}+ \Theta \gamma{}_{\bmb\bme} d{}^\bme{}_{\bmd\bmc\bmzero} +
\Theta \gamma{}_{\bmc\bme} d{}^\bme{}_{\bmd\bmzero\bmb} - \frac{1}{2}\epsilon{}_{\bmzero\bmb\bmc}{}^\bmf
\epsilon{}_\bmd{}^{\bme\bmg\bmh} \Lambda{}_{\bmf\bmg\bmh} \beta_\bme, \\
&& \hat{\nabla}{}_\bmzero \hat{\Omega}{}_{\bmb\bmc} = \hat{\Xi}{}^\bme{}_{[\bmb}{}^{\bma\bmf}
d{}_{\bmc]\bme\bmf\bma} - \hat{\Xi}{}^\bme{}_\bmf{}^{\bma\bmf} d{}_{\bme\bma\bmb\bmc} + \frac{1}{2}
\hat{\Sigma}{}_\bma{}^\bme{}_\bmf \nabla_\bme d{}^{\bmf\bma}{}_{\bmb\bmc} +
\varsigma{}^{\bmf\bma}d{}_{\bmf\bma\bmb\bmc}-\chi \Omega{}_{\bmb\bmc},\\
&&\hat{\nabla}_\bmzero \delta_\bmi= \gamma_{\bmi\bmzero}- \hat{\Gamma}{}_\bmi{}^\bme{}_\bmzero \delta_\bme; \\
&& \hat{\nabla}_\bmzero \gamma_{\bmi\bmc}= - \gamma_{\bmj\bmc} \hat{\Gamma}{}_\bmi{}^\bmj{}_\bmzero
   - \beta_\bmzero \gamma_{\bmi\bmc} - \beta_\bmc \gamma_{\bmi\bmzero} + \eta_{\bmzero\bmc}(\beta^\bme
   \gamma_{\bmi\bme} - 2 \lambda \Theta^{-2} \delta_\bmi), \\
&& \hat{\nabla}_\bmzero \varsigma_{\bmj\bmk} =
   \hat{\Gamma}{}_\bmj{}^\bme{}_\bmzero \varsigma{}_{\bmk\bme} +
   \Gamma{}_\bmk{}^\bme{}_\bmzero \varsigma{}_{\bme\bmj} + \frac{1}{2}
   \hat{\Delta}_{\bmj\bmk\bmzero} + \frac{1}{2}
   \hat{\Xi}{}^\bme{}_{\bmzero\bmj\bmk} f_\bme + \frac{1}{2}
   \hat{\Sigma}{}_\bmj{}^\bme{}_\bmk
   \hat{\Gamma}{}_\bme{}^\bmf{}_\bmzero f_\bmf, \label{SubsidaryEqnLast}
\end{eqnarray}
\end{subequations}
is that the zero-quantities
$\hat{\Sigma}_\bma{}^\bmc{}_\bmb$, $\hat{\Xi}^\bma{}_{\bmb\bmc\bmd}$,
$\hat{\Delta}_{\bma\bmb\bmc}$, $\hat{\Lambda}_{\bma\bmb\bmc}$,
$\delta_{\bma\bmb}$, $\gamma_{\bma\bmb}$ and $\varsigma_{\bma\bmb}$
satisfy, if the conformal evolution equations
\eqref{XCFEEvolution1}-\eqref{XCFEEvolution5} hold, a symmetric
hyperbolic system which is homogeneous in the zero-quantities
---accordingly, the particular situation in which all the
zero-quantities vanish identically gives rise to the subsidiary
evolution system.  The subsidiary system is regular away from the
conformal boundary ---i.e. the sets for which the conformal factor
vanishes. 

\section{Initial data for the evolution equations}
\label{Section:InitialData}

Given a solution $(S,\tilde{\bmh},\tilde{\bmK})$ to the Einstein
constraint equations (i.e the Hamiltonian and the momentum
constraints), there exists an algebraic procedure to compute initial
data for the conformal evolution equations ---see e.g. \cite{CFEBook},
Lemma 11.1, page 265. Now, a suitable perturbative existence theorem
which covers perturbations of the initial data implied 
by the metric \eqref{BackgroundPhysicalMetric} on the hypersurfaces of
constant $t$ has been provided in \cite{ValWil20} ---see Theorem
1. From this result one can deduce the following assertion:

\begin{proposition}
\label{Proposition:ExistenceID}
Let $(\mathcal{S},\mathring{\bmh},\mathring{\bmK})$ with
$\mathcal{S}$ compact,  $\mathring{\tilde{\bmh}}$ a smooth Riemannian metric
of constant negative curvature and $\mathring{\bmK}=\varkappa
\mathring{\bmh}$ with $\varkappa$ a constant, denote a initial data set for the vacuum Einstein field
equations with positive Cosmological constant. Then for each pair of
sufficiently small (in the sense of suitable Sobolev norms) tensors
$T_{ij}$ and $\bar{T}_{ij}$ over $\mathcal{S}$, transverse-tracefree
with respect to $\mathring{\bmh}$, and each sufficiently small scalar
field $\phi$ over $\mathcal{S}$, there exists a solution of the
Einstein constraint equations $(\mathcal{S},\bmh,\bmK)$ with positive
Cosmological constant which is suitably close to
$(\mathcal{S},\mathring{\bmh},\mathring{\bmK})$ and such that 
$\mbox{tr}_{\mathring{\bmh}} (\bmK-\mathring{\bmK})=\phi$ and for which
the electric and magnetic parts of the Weyl tensor (restricted to
$\mathcal{S}$) of the resulting spacetime development take the form 
\begin{eqnarray*}
&& d_{ij} = \mathring{L}(\bmX)_{ij} +T_{ij}
   -\frac{1}{3}\mbox{tr}_{\bmh}(\mathring{L}(\bmX) +\bmT) h_{ij}, \\
&& d^*_{ij} = \mathring{L}(\bar\bmX)_{ij} +\bar T_{ij}
   -\frac{1}{3}\mbox{tr}_{\bmh}(\mathring{L}(\bar\bmX) +\bar\bmT) h_{ij},
\end{eqnarray*}
for some covectors $\bmX$, $\bar\bmX$ over $\mathcal{S}$ and where
$\mathring{L}$ denotes the conformal Killing operator with respect to
$\mathring{\bmh}$. 
\end{proposition}

\begin{remark}
{\em Thus, choosing the free data $T_{ij}$, $\bar{T}_{ij}$ and $\phi$
  suitably small one can ensure that the \emph{perturbed data}
  $(\mathcal{S},\bmh,\bmK)$ is as close to
  $(\mathcal{S},\mathring{\bmh},\mathring{\bmK})$. Accordingly, the
  associated initial data for the conformal evolution equations will
  be close to initial data for the background solution.}
\end{remark}

\begin{remark}
{\em Theorem 1 in \cite{ValWil20} applies to the broader class of
  \emph{conformally rigid hyperbolic} compact manifolds ---that is,
  Einstein manifolds with negative Ricci scalar which do not admit a
  non-trivial Codazzi tensor; see the discussion in Section 3.4.3 of
  this reference. The precise statement of the result also excludes
  values of $\varkappa$ which are related in a specific manner to the
  eigenvalues of the Laplacian of $\mathring{\bmh}$ ---however, we
  do not require this level of detail in the subsequent discussion.} 
\end{remark}

\section{Analysis of the existence and stability of solutions} 
\label{Section:Existence}
In this section we develop the theory of the existence, uniqueness and
stability of solutions to the Einstein field equations which can be
regarded as perturbations of the background solution. The argument
proceeds in several steps: first, the Cauchy stability of solutions to
symmetric hyperbolic systems is used to conclude the existence of
solutions to the conformal evolution system
\eqref{XCFEEvolution1}-\eqref{XCFEEvolution6}; in a second step the
uniqueness of solutions to the subsidiary system
\eqref{SubsidaryEqnFirst}-\eqref{SubsidaryEqnLast} to argue the
propagation of constraints; finally general theory of the conformal
Einstein field equations is invoked to establish the connection
between solutions to the conformal equations and actual solutions to
the Einstein field equations.

\subsection{A symmetric hyperbolic evolution system}
In the following we look for solutions to the system
\eqref{he1}-\eqref{he2} of the form 
\[
\hat{\mathbf{u}} = \mathring{\mathbf{u}} + \breve{\mathbf{u}} 
\]
where $\mathring{\mathbf{u}}$ is the solution to the conformal
evolution equations \eqref{XCFEEvolution1}-\eqref{XCFEEvolution6}
implied by a background solution, while $\breve{u}$ denotes a (small)
perturbation. Accordingly, making use of the schematic notation of
equations \eqref{he1}-\eqref{he2} one can set
\begin{subequations}
\begin{eqnarray} 
&& \label{per1} \hat{\bmupsilon} = \mathring{\bmupsilon} + \breve{\bmupsilon}, \hspace{1cm} \hat{\bmphi} = \breve{\bmphi}, \\
&&  \label{per2} \hat{\bme} = \mathring{\bme} + \breve{\bme}, \hspace{1cm}  \hat{\bmGamma} = \mathring{\bmGamma} + \check{\bmGamma}.
\end{eqnarray}
\end{subequations}  
Now, we have found that on the initial surface $\mathcal{S}_\star$
described by the condition $\tau=\tau_\star$ one can write
$\mathring{\mathbf{u}}_\star =( \mathring{\bmupsilon}_\star,
\mathring{\bmphi}_\star) = ( \mathring{\bmupsilon}_\star, 0)$. As the
conformal factor $\Theta$ and the covector $\bmd$ are
universal, it follows that
\[
 \partial_\tau \mathring{\bmupsilon} = \mathbf{K} \mathring{\bmupsilon} + \mathbf{Q}( \mathring{\bmupsilon}, \mathring{\bmupsilon}). 
\]
Substituting \eqref{per1} and \eqref{per2} into equations \eqref{he1}
and \eqref{he2} yields evolution equations for
$\breve{\mathbf{u}}=(\breve{\bmupsilon}, \breve{\bmphi})$ which,
schematically, take the form 
\begin{subequations}
\begin{eqnarray}
&& \label{he3}\partial_\tau \breve{\bmupsilon}= \mathbf{K} \breve{\bmupsilon} + \mathbf{Q}(\mathring{\bmGamma}+ \breve{\bmGamma}) \breve{\bmupsilon}+ \mathbf{Q}(\breve{\bmGamma})\mathring{\bmupsilon} + \mathbf{L}(\bar{x}) \breve{\bmphi},  \\
&&\label{he4} (\mathbf{I} + \mathbf{A}^0(\mathring{\bme} + \breve{\bme})) \partial_\tau \breve{\bmphi} + \mathbf{A}^\alpha (\mathring{\bme} + \breve{\bme}) \partial_\alpha \breve{\bmphi} =\mathbf{B}(\mathring{\bmGamma}+ \breve{\bmGamma}) \breve{\bmphi}.
\end{eqnarray}
\end{subequations}
Now, in the following it is convenient to define
\[
\bar{\mathbf{A}}^0(\tau, \underline{x}, \breve{\mathbf{u}}) \equiv \begin{pmatrix*}
      {\rm I} &  0 \\
      0 & {\rm I} + \mathbf{A}^0(\mathring{\bme}+ \breve{\bme})
     \end{pmatrix*},
\qquad 
\bar{\mathbf{A}}^\alpha(\tau, \underline{x}, \breve{\mathbf{u}}) \equiv \begin{pmatrix*}
  0 &   0 \\
      0 &   \mathbf{A}^\alpha(\mathring{\bme}+ \breve{\bme})
     \end{pmatrix*}
\]
and
\[
\bar{\mathbf{B}}(\tau, \underline{x}, \breve{\mathbf{u}})\equiv \breve{\mathbf{u}}\bar{\mathbf{Q}}\breve{\mathbf{u}}+ \bar{\mathbf{L}}(\bar{x})\breve{\mathbf{u}}+ \bar{\mathbf{K}}\breve{\mathbf{u}},
\]
where 
\[
\breve{\mathbf{u}}\bar{\mathbf{Q}}\breve{\mathbf{u}} \equiv \begin{pmatrix*}
    \breve{\bmupsilon}\mathbf{Q}\breve{\bmupsilon} &     0 \\
     0 &
     \mathbf{B}(\breve{\bmGamma})\breve{\bmphi}
     \end{pmatrix*},
\qquad 
\bar{\mathbf{L}}(\bar{x})\breve{\mathbf{u}} \equiv \begin{pmatrix*}
  \mathring{\bmupsilon}\mathbf{Q}\breve{\bmupsilon}+ \mathbf{Q}(\breve{\bmGamma})\mathring{\bmupsilon}&
    \mathbf{L}(\bar{x})\breve{\bmphi} \\
      0 &  0
     \end{pmatrix*},
\qquad 
\bar{\mathbf{K}} \breve{\mathbf{u}}\equiv \begin{pmatrix*}
 \mathbf{K} \breve{\bmupsilon} &  0  \\
      0 & \mathbf{B}(\mathring{\Gamma})\breve{\bmphi}
     \end{pmatrix*},
\]
denote, respectively, quadratic, linear and constant terms in the unknowns. In terms of
the latter it is possible to rewrite the system \eqref{he3} and
\eqref{he4} in the form
\begin{equation}
 \bar{\mathbf{A}}^0(\tau, \underline{x}, \breve{\mathbf{u}})\partial_\tau \breve{\mathbf{u}}+ \bar{\mathbf{A}}^\alpha(\tau, \underline{x}, \breve{\mathbf{u}})\partial_\alpha \breve{\mathbf{u}}=\bar{\mathbf{B}}(\tau, \underline{x}, \breve{\mathbf{u}}). \label{he5}
\end{equation}
From the discussion in the previous sections, it follows that the
system described by \eqref{he5} is a symmetric hyperbolic system for
which the theory of \cite{Kat75} can be applied. The natural domain
of the solutions to this system is of the form
\[
\mathcal{M} = [\tau_\star, \tau_\bullet)\times \mathcal{S}, \qquad
\tau_\star\in(0,1), \qquad \tau_\bullet\geq 1.
\]

\subsection{The existence, uniqueness and Cauchy stability of the solution}
The existence of de Sitter-like solutions to the conformal evolution
system \eqref{he5} is given by the following proposition:

\begin{proposition} [\textbf{\em existence and uniqueness of the
    solutions to the perturbed de Sitter-like evolution equations}]
\label{Proposition:ExistenceConformalEvolution}
 Given $\mathbf{u}_\star=\mathring{\mathbf{u}}_\star
 +\breve{\mathbf{u}}_\star $ 
 and  $m \geq 4$, one has that:
\begin{itemize}
\item[(i)] There exists $\varepsilon >0$ such that if 
\begin{equation}
||\breve{\mathbf{u}}_\star||_{\mathcal{S},m} < \varepsilon, \label{SizeData}
\end{equation}
then there exists a unique solution $\breve{\mathbf{u}}\in C^{m-2}(\big{[}  \tau_\star, \frac{3}{2} \big{)}  \times \mathcal{S}, \mathbb{R}^N)$ to the Cauchy problem for the
conformal evolution equations \eqref{he5} with initial data
$\mathbf{u}(0,\underline{x})=\breve{\mathbf{u}}_\star$, $\tau_\star>0$ and with
$N$ denoting the dimension of the vector $\mathbf{u}$.

\item[(ii)] Given a sequence of initial data $
  \breve{\mathbf{u}}{}^{(n)}_\star$ such
  that 
\[
||\breve{\mathbf{u}}{}^{(n)}_\star||_{\mathcal{S},m} < \varepsilon,
\qquad \mbox{and} \qquad
||\breve{\mathbf{u}}{}^{(n)}{}_\star||_{\mathcal{S},m}
\xrightarrow{n\rightarrow\infty} 0, 
\]
then for the corresponding solutions $\breve{\mathbf{u}}{}^{(n)} \in C^{m-2} \big{(}\big{[}  \tau_\star, \frac{3}{2} \big{)} \times \mathcal{S}, \mathbb{R}^N \big{)}$, one has $|| \breve{\mathbf{u}}{}^{(n)}||_{\mathcal{S},m} \rightarrow 0$ uniformly in $\tau \in \big{[}  \tau_\star, \frac{3}{2} \big{)}$ as $n \rightarrow \infty$.
\end{itemize}

\end{proposition}

\begin{remark}
{\em In the above proposition
$||\breve{\mathbf{u}}_\star||_{\mathcal{S},m}$ denotes the standard
$L^2$-Sobolev norm over $\mathcal{S}$ of order $m\geq 4$ of the
independent components of the vector $\breve{\mathbf{u}}_\star$. }
\end{remark}

\begin{proof}
The proof is an application of the existence and stability results for
symmetric hyperbolic systems with compact spatial sections ---see
e.g. \cite{CFEBook}, Section 12.3 which, in turn, follows from Kato's
theory for symmetric hyperbolic systems over $\mathbb{R}^n$
\cite{Kat75}. More precisely, since the 3-dimensional manifold
$\mathcal{S}$ is compact, there exists a finite cover consisting of
open sets $\mathcal{R}_1, \dots, \mathcal{R}_M \subset \mathcal{S}$
such that $ \cup{}_{i=1}{}^M \mathcal{R}{}_i = \mathcal{S}$. On each
of the open sets $\mathcal{R}{}_i$ it is possible to introduce
coordinates $\underline{x}{}_i \equiv (x{}^\alpha{}_i)$ which allow
one to identify $\mathcal{R}{}_i$ with open subsets $\mathcal{B}_i
\subset \mathbb{R}^3$. As $\mathcal{S}$ is assumed to be a smooth
manifold, the coordinate patches can be chosen so that the change of
coordinates on intersecting sets is smooth. The initial data
$\breve{\mathbf{u}}{}_\star: \mathcal{S} \rightarrow \mathbb{R}^N$ is a smooth
function on $\mathcal{S}$ and can be restricted to a particular open
set $\mathcal{R}{}_i$. The restriction $\breve{\mathbf{u}}{}_{i\star}$, in
local coordinates $x_i$, can be regarded as a function
$\breve{\mathbf{u}}_{i\star}: \mathcal{B}_i \rightarrow \mathbb{R}^N$. Now,
assuming that $\mathcal{R} \subset \mathbb{R}^3$ is bounded with
smooth boundary $\partial \mathcal{R}$, it is possible to extend $\breve{\mathbf{u}}_{i\star}$
to a function $\mathcal{E} \breve{\mathbf{u}}_{i\star}: \mathbb{R}^3
\rightarrow \mathbb{R}^N$ ---see e.g. Proposition
12.2 in \cite{CFEBook}. Using these extensions it is possible to define the Sobolev norm 
\[
|| \breve{\mathbf{u}}_\star||_{\mathcal{S}, m} \equiv \sum\limits_{i=1}^M
||\breve{\mathbf{u}}{}_{i\star}||_{\mathbb{R}^3, m}. 
\]

Now, for each of the  $\mathcal{E} \breve{\mathbf{u}}_{i\star}$ one can
formulate an initial value problem of the form
\begin{eqnarray*}
&& \bar{\mathbf{A}}^0(\tau, \underline{x}, \breve{\mathbf{u}}) \partial_\tau \breve{\mathbf{u}} + \bar{\mathbf{A}}^\alpha (\tau, \underline{x}, \breve{\mathbf{u}}) \partial_\alpha \breve{\mathbf{u}} = \mathcal{B}(\tau, \underline{x}, \breve{\mathbf{u}}), \\
&& \breve{\mathbf{u}}(0, \underline{x})=\mathcal{E} \breve{\mathbf{u}}_{i\star}( \underline{x}) \in H{}^m (\mathcal{S}, \mathbb{R}{}^N) \hspace{0.5cm} {\rm for} \; m \geq 4.
\end{eqnarray*}
For this initial value problem it is observed that:

\begin{itemize}
\item[(a)] The matrices $\bar{\mathbf{A}}^\mu (\tau, \underline{x},
  \mathcal{E}\breve{\mathbf{u}}_{i\star})$ are positive definite and depend
  linearly on the solution $\breve{\mathbf{u}}_i$ with coefficients which
  are constant.

\item[(b)] The dependence of $\mathcal{B}$ on $\breve{\mathbf{u}}_i$ is at most
  quadratic: there are linear and quadratic terms for the connection
  coefficients; linear terms for the components of the Schouten
  tensor. The explicit dependence on $(\tau,\underline{x})$ comes from
  the conformal factor and the covector $d_a$ ---this dependence is
  smooth.

\item[(c)]  The connection coefficients and the components of the Schouten
  tensor of the background solution are smooth functions ($C^\infty$)
  of $(\tau,\underline{x})$.

\item[(d)] The dependence of the frame coefficients of the background
  solution is smooth ($C^\infty$) on $\tau$ for $\tau\in[\tau_\star,\tfrac{3}{2}]$ with
  $\tau_\star\neq 0$. 
\end{itemize}

It follows from the above observations that the system considered in
the present article satisfies the conditions of Kato's theorems ---see
Appendix \ref{Appendix:KatoThm}. This theory implies existence,
uniqueness and stability ---i.e. points (i) and (ii) in the theorem. Notice, however, that strictly speaking,
this theorem only applies to settings in which the spatial sections
are diffeomorphic to $\mathbb{R}^3$. To address this one makes use of
the following strategy: standard results on causality theory  imply that
\[ 
D^+(\mathcal{R}_i) \cap I^+(\mathcal{S} \setminus
\mathcal{R}_i)= \emptyset, 
\]
where $D^+(\mathcal{R}_i)$ denotes the causal future of $\mathcal{R}_i$ ---see
e.g. \cite{CFEBook}, Theorem 14.1. Accordingly, the value of
$\breve{\mathbf{u}}$ on $\mathcal{D}_i\equiv D^+(\mathcal{R}_i)$ is determined only by
the data on $\mathcal{R}_i$. Then the solution on
$\mathcal{D}_i$ is independent of the particular
extension $\mathcal{E} \breve{\mathbf{u}}_{i\star}$ being used. Hence, one
can speak of a solution $\breve{\mathbf{u}}_i$ on a domain $\mathcal{D}_i
\subset [\tau_\star, \tau_i] \times \mathcal{R}_i$. Since the manifold
is smooth and as a consequence of uniqueness, it follows that given
two solutions $\breve{\mathbf{u}}_i$ and $\breve{\mathbf{u}}_j$ defined,
respectively, on intersecting domains $\mathcal{D}_i$ and
$\mathcal{D}_j$ they must coincide on $\mathcal{D}_i \bigcap
\mathcal{D}_j$. Proceeding in the same manner over the whole finite
cover of $\mathcal{S}$ and since the compactness of $\mathcal{S}$
ensures the existence of a minimum non-zero existence time for the
whole of the domains $\mathcal{D}_i$, then there is a unique solution
$\breve{\mathbf{u}}$ on $[\tau_\star, \tfrac{3}{2}] \times \mathcal{S}$ with
$\frac{3}{2}={\rm min}{}_{i=1, \dots, M} \{\tau_i
\}$ which is constructed by patching together the localised solutions
$\breve{\mathbf{u}}_1, \dots, \breve{\mathbf{u}}_M$ defined, respectively on the
domains $ \mathcal{D}_i, \dots, \mathcal{D}_M$. The existence interval
$[\tau_\star,\tfrac{3}{2})$ follows from the fact that the background
solution $\mathring{\mathbf{u}}$ has this existence interval. 

\end{proof}

\begin{remark} 
{\em 
The existence and Cauchy stability of the solution to the initial
value problem for the original conformal evolution problem
\begin{eqnarray*} 
&& \mathbf{A}^0(\tau, \underline{x},
\hat{\mathbf{u}})\partial_\tau \hat{\mathbf{u}}+ \mathbf{A}^\alpha(\tau,
\underline{x}, \hat{\mathbf{u}})\partial_\alpha \hat{\mathbf{u}}=\mathbf{B}(\tau,
\underline{x}, \hat{\mathbf{u}}), \\ 
&& \hat{\mathbf{u}}|_\star=
\mathring{\mathbf{u}}_\star + \breve{\mathbf{u}}_\star \in H^m (\mathcal{S},
\mathbb{R}^N) \quad {\rm for} \quad m \geq 4
\end{eqnarray*} 
follows from the fact that $\hat{\mathbf{u}}$ satisfies the
same properties as $\breve{\mathbf{u}}$ in Proposition
\ref{Proposition:ExistenceConformalEvolution} and then it exists in
the same solution manifold and with the same regularity properties,
existence and uniqueness. }
 \end{remark}

\subsection{Propagation of the constraints}
In this section we discuss the so-called \emph{propagation of the
  constraints}. This argument is essential to establish the connection
between solutions to the conformal evolution systems and actual
solutions to the Einstein field equations. More precisely, one has the following:

\begin{proposition}[\textbf{\em propagation of the constraints}]
 Let $\hat{\mathbf{u}}_\star = \mathring{\mathbf{u}}_\star + \breve{\mathbf{u}}_\star$ denote
 initial data for the conformal evolution equations on a $3$-manifold
 $\mathcal{S}_\star\approx \mathcal{S}$ such that 
\[
\hat{\Sigma}{}_\bma{}^\bmc {}_\bmb |_{\mathcal{S}_\star}=0, \quad
\hat{\Xi}{}^c {}_{\bmd\bma\bmb}  |_{\mathcal{S}_\star}=0, \quad
\hat{\Delta}{}_{\bma\bmb\bmc}  |_{\mathcal{S}_\star}=0, \quad  \hat{
  \Lambda}_{\bma\bmb\bmc} |_{\mathcal{S}_\star}=0, 
\]
and
\[
 \delta{}_\bma |_{\mathcal{S}_\star}=0, \quad \gamma{}_{\bma\bmb}
 |_{\mathcal{S}_\star}=0, \quad \varsigma{}_{\bma\bmb} |_{\mathcal{S}_\star}=0, 
\]
then the solution $\breve{\mathbf{u}}$ to the conformal evolution equations
given by Proposition \ref{Proposition:ExistenceConformalEvolution}
implies a $C^{m-2}$ solution $\hat{\mathbf{u}}= \mathring{\mathbf{u}} +
\breve{\mathbf{u}}$ to the extended conformal field equations on $\big{[}
\tau_\star, 1 \big{)} \times \mathcal{S}$.
\end{proposition}

\begin{proof}
The proof follows from the properties of the subsidiary evolution
system. First, it is observed that by assumption
\[
 \hat{\Sigma}{}_\bmzero{}^\bmc{}_\bmb=0, \qquad  \hat{\Xi}{}^\bmc{}_{\bmd\bmzero\bmb}=0,
 \qquad   \hat{\Delta}{}_{\bmzero\bmb\bmc}=0, 
\]
hold ---cfr. the equations in \eqref{ecfe6}. Moreover, the associated evolution equations are expressed in
terms of a conformal Gaussian gauge system and the independent
components of the rescaled Weyl tensor satisfy either the evolution
system \eqref{EvolutionElectricWeyl} and
\eqref{EvolutionMagneticWeyl}. Now, following the discussion of
Section \ref{Subsection:SubsidiarySystem}, the independent components of the zero-quantities
\[
 \hat{\Sigma}{}_\bma{}^\bmc{}_\bmb, \quad  \hat{\Xi}{}^\bmc{}_{\bmd\bma\bmb},
 \quad   \hat{\Delta}{}_{\bma\bmb\bmc}, \quad
 \hat{\Lambda}{}_{\bma\bmb\bmc}, \quad  \delta_{\bma}, \quad
 {\gamma}{}_{\bma\bmb}, \quad   {\varsigma}{}_{\bma\bmb},
\]
which are not determined by either the evolution equations or gauge
conditions satisfy a symmetric hyperbolic system which is homogeneous
in the zero-quantities. More precisely, defining
$\hat{\mathbf{X}}\equiv(\hat{\Sigma}{}_\bma{}^\bmc{}_\bmb,\hat{\Xi}{}^\bmc{}_{\bmd\bma\bmb},
\hat{\Delta}{}_{\bma\bmb\bmc}, \hat{\Lambda}{}_{\bma\bmb\bmc},
\delta_{\bma}, {\gamma}{}_{\bma\bmb}, {\varsigma}{}_{\bma\bmb})$,
these equations can be recasted as a symmetric hyperbolic system of
the form
\[
\partial_\tau \hat{\mathbf{X}}= \mathbf{H}(\hat{\mathbf{X}}),
\]
where $\mathbf{H}$ is a homogeneous function of its arguments
---i.e. $\mathbf{H}(\bmzero)=\bmzero$. 
It follows then that a solution to the initial value problem
\begin{eqnarray*}
&& \partial_\tau \hat{\mathbf{X}}= \mathbf{H}(\hat{\mathbf{X}}), \\
&& \hat{\mathbf{X}}_\star =0.
\end{eqnarray*}
is given (trivially) by $\hat{\mathbf{X}}=0$. Moreover, from
Kato's theorem it follows that this is the unique solution. Thus, the
zero-quantities must vanish on $\big{[}\tau_\star, 1\big{)}\times
\mathcal{S}$. That is, the solution $\breve{\mathbf{u}}$ to the conformal
evolution equations implies a solution to the extended conformal
Einstein field equations over the latter domain. 
\end{proof}

From the above statement, making use of the relation between the
extended conformal Einstein field equations and the actual Einstein
field equations ---see Proposition 8.3 in 208 --- in \cite{CFEBook} it follows
the following:

\begin{corollary}
\label{Corollary:SolutionsEFE}
The metric 
\[
\tilde{\bmg} = \Theta^{-2} \bmg
\]
obtained from the solution to the conformal evolution equations given
in Proposition \ref{Proposition:ExistenceConformalEvolution} implies a
solution to the vacuum Einstein field equations with $\lambda=3$.
\end{corollary}

\section{Future geodesic completeness}
\label{Section:GeodesicCompleteness}

In this section we discuss the future geodesic completeness of the
spacetimes obtained in the previous section. Our analysis
distinguishes two cases: null geodesics and timelike geodesics. 

%To prove this, following the same approach used in \cite{LueVal13a}, we consider the following cases

\subsection{Null geodesics} 
As a consequence of the compactness of the unphysical manifold
\[
\mathcal{M}= \bigg{\{} (\tau, \underline{x}) \in \mathbb{R} \times
\mathcal{S} \; | \; \tau_{\bullet} \leq \tau \leq1 \bigg{\}}, 
\]
null geodesics in the unphysical manifold starting at the initial
hypersurface $\mathcal{S}_\star$, reach the conformal boundary in a finite
amount of affine parameter. Furthermore, null geodesics with respect
to the unphysical metric $\bmg$ coincide, up to a reparametrisation,
with null geodesics respect to the physical metric $\tilde{\bmg}$ on
$\tilde{\mathcal{M}}$. More precisely, let $\gamma$ be a null geodesic
in $(\mathcal{M},\bmg)$ with affine parameter $v$ such that $v=0$ on
$\partial \tilde{\mathcal{M}}$. The equations for  $\gamma$ are
\[
\frac{ \mathrm{d}^2 x^\mu}{\mathrm{d} v^2} + \Gamma{}^\mu{}_{\nu\lambda} \frac{ \mathrm{d} x^\nu}{\mathrm{d} v} \frac{\mathrm{d}
  x^\lambda}{\mathrm{d} v}=0. 
\]
Let $\tilde{\gamma}$ denote the corresponding geodesics in $\tilde{\mathcal{M}}$. Using a different parameter $\tilde{v}=\tilde{v}(v)$ and the relation between the Christoffel symbols $\Gamma{}^\mu{}_{\nu\lambda}$ and $\tilde{\Gamma}{}^\mu{}_{\nu\lambda}$ 
it follows that
\[ 
\frac{\mathrm{d}^2 x^\mu}{\mathrm{d} \tilde{v}^2} + \tilde{\Gamma}{}^\mu{}_{\nu\lambda} \frac{\mathrm{d}
  x^\nu}{\mathrm{d} \tilde{v}} \frac{\mathrm{d} x^\lambda}{d \tilde{v}} = - \frac{1}{
  \tilde{v}'} \bigg{(}\frac{ \tilde{v}''}{ \tilde{v}'}+ 2
\frac{\Theta'}{\Theta} \bigg{)} \frac{\mathrm{d} x^\mu}{\mathrm{d} \tilde{v}}. 
\]
By requiring that $\tilde{v}$ to be an affine parameter the right hand
side must vanish. This implies $\tilde{v}'= {\rm const}/\Theta^2$, and
absorbing the constant into $\tilde{v}$ we obtain
\[
\frac{ \mathrm{d} \tilde{v}}{\mathrm{d} v}= \frac{1}{\Theta^2}. 
\]
Furthermore, at $\mathscr{I}^+$, $\Theta=0$ and $\mathbf{d} \Theta
\neq 0$, and we may choose $v$ so that near $\partial \tilde{M}$, $v
\sim - \Theta$. Thus $\tilde{v} \sim - 1/v$ becomes unbounded
---i.e. the physical affine parameter for the physical geodesic must
blow up as $\Theta \rightarrow 0$. Thus, $\tilde{\gamma}$ never
reaches $\partial \tilde{\mathcal{M}}$ and the null geodesic must be
complete   ---see also the discussion in \cite{Ste91}, Chapter 3.

\subsection{Timelike geodesics} 
The argument used for null geodesics cannot readily be applied to the
discussion of timelike geodesics as these are not conformally
invariant. Instead, we make use of timelike conformal geodesics.

\medskip
Every timelike metric geodesic on the physical spacetime
$(\tilde{\mathcal{M}}, \tilde{\bmg})$ can be recast, after a
reparametrisation, as a conformal geodesics $(x, \tilde{\bmbeta})$
---see e.g. \cite{CFEBook}, Lemma 5.2 in page 131; also
\cite{FriSch87}. Under the rescaling $\bmg= \Theta^2 \tilde{\bmg}$,
the conformal geodesic $(x, \tilde{\bmbeta})$ transforms into a
geodesic $(x, \beta)$ in the unphysical spacetime $(\mathcal{M},
\bmg)$. Now, it is known that any $\bmg$-conformal geodesic that
leaves $\mathscr{I}^+$ orthogonally into the past, is up to a
reparametrisation, a timelike future complete geodesic for the
physical metric $\tilde{g}$ ---see
e.g. \cite{FriSch87,Fri17}. Moreover, a conformal geodesic through a
point of $\mathscr{I}^+$ which is not orthogonal to the conformal
boundary cannot represent a geodesic in the physical spacetime. 

Now, from the $\tilde{\bmg}$-future geodesic completeness of the
background solution (see Appendix \ref{Appendix:GeodesicCompleteness})
it follows that every conformal geodesic in the background spacetime
starting orthogonal to the initial hypersurface $\mathcal{S}_\star$
must reach the conformal boundary $\mathscr{I}^+$. Hence, every
timelike $\tilde{\bmg}$-geodesic is, up to a reparametrisation, a
timelike conformal curve reaching $\mathscr{I}^+$
orthogonally. Moreover, let us consider a pair
$(x(\tau),\tilde\bmbeta(\tau))$ with parameter
$\tau\in\mathbb{R}$. Furthermore, let us suppose that this geodesic
starts at $\tau=\tau_\star$, i.e. the initial hypersurface
$\mathcal{S}$, and it reaches the conformal boundary $\mathscr{I}^+$
at $\tau=1$. Now, consider a small perturbation of the quantities $(x,
\tilde\bmbeta)$ so that
\begin{eqnarray*} 
&&\hat{x}= x + \breve{x}, \\
&& \hat{\bmbeta}= \tilde\bmbeta + \breve\bmbeta,
\end{eqnarray*}
where $ \breve{x}$ and $\breve\bmbeta$ are small perturbations. In
this case, the perturbed conformal geodesic equations read 
\begin{eqnarray*}
&& \tilde{\nabla}_{\bmx'} (\bmx'+ \breve{\bmx}') = -2 \langle (\tilde\bmbeta + \breve\bmbeta),(\bmx'+ \breve{\bmx}')\rangle
   (\bmx'+ \breve{\bmx}') + \tilde{\bmg}((\bmx'+ \breve{\bmx}'),(\bmx'+ \breve{\bmx}')) (\tilde\bmbeta + \breve\bmbeta)^\sharp, \\
&& \tilde{\nabla}_{\bmx'} (\tilde\bmbeta + \breve\bmbeta) = \langle (\tilde\bmbeta + \breve\bmbeta),(\bmx' + \breve\bmx')\rangle
   (\tilde\bmbeta + \breve\bmbeta) - \frac{1}{2}\bmg^\sharp ((\tilde\bmbeta + \breve\bmbeta),(\tilde\bmbeta + \breve\bmbeta))
   (\bmx + \breve\bmx)^{\prime\flat} + \tilde{\bmL}((\bmx'+ \breve\bmx'),\cdot),
\end{eqnarray*}
where the metric, covariant derivative and Schouten tensor are those
obtained from the solution to the Einstein field equations given in
Corollary \ref{Corollary:SolutionsEFE}.  These equations can be read as a system of ordinary differential
equations for the fields $ \breve{x}$ and $\breve\bmbeta$. Because of
the smoothness of the perturbed spacetime it follows that one can make
use of the stability theory for ordinary differential equations ---see
e.g. \cite{Har87}, Theorem 2.1 in page 94 and Corollary 4.1 in page
101. In particular, these conformal geodesics will have the same
existence interval as those in the background spacetime. Accordingly,
it follows that $(\tilde{\mathcal{M}}, \tilde{\bmg})$ is future $\tilde{\bmg}$-geodesically complete.

\begin{remark}
{\em An alternative way of concluding the future geodesic completeness
of the solutions to the Einstein field equations provided by Corollary
\ref{Corollary:SolutionsEFE} is to make use of the theory in
\cite{ChoCot02} ---see also Appendix
\ref{Appendix:GeodesicCompleteness}. By choosing the $\varepsilon>0$
in condition \eqref{SizeData} of Proposition
\ref{Proposition:ExistenceConformalEvolution} sufficiently small, it
can be shown that the physical metric $\tilde\bmg$ satisfies the
bounds required to shown geodesic completeness.  
}
\end{remark}

\section{The main result}
\label{Section:MainTheorem}

We summarise the discussion of the preceding sections with a more
detailed formulation of the main result of this article:

\begin{theorem}
\label{Theorem:Main}
Let $\hat{\mathbf{u}}_\star =\mathring{\mathbf{u}}_\star +\breve{\mathbf{u}}_\star$
denote smooth initial data for the conformal evolution equations satisfying
the conformal constraint equations on a hypersurface
$\mathcal{S}_\star$. Then, there exists $\varepsilon>0$ such that if
\[
||\breve{\mathbf{u}}_\star||_{\mathcal{S}_\star,m} < \varepsilon, \qquad
m\geq 4
\]
then there exists a unique $C^{m-2}$ solution $\tilde{\bmg}$ to the vacuum Einstein field equation with
positive Cosmological constant over $[\tau_\star,\infty)\times
\mathcal{S}_\star$ for $\tau_\star>0$ which is future geodesically
complete and whose restriction to $\mathcal{S}_\star$ implies the
initial data $\hat{\mathbf{u}}_\star$. Moreover, the solution
$\hat{\mathbf{u}}$ remains suitably close (in the Sobolev norm
$\parallel \cdot \parallel_{\mathcal{S},m}$) to the background solution $\mathring{\bmu}$.
\end{theorem}

\begin{remark}
{\em It follows from Proposition \ref{Proposition:ExistenceID} that
  there exists an open set of initial data for the Einstein field
  equations satisfying the hypothesis of the above theorem.}
\end{remark}

\appendix

\section{On Kato's existence and stability result for symmetric
  hyperbolic systems}
\label{Appendix:KatoThm}

In this appendix we make some remarks concerning the hypothesis in
Kato's existence, uniqueness and stability result for symmetric
hyperbolic equations in \cite{Kat75}. The results in this reference
and, in particular the main Theorem II, are very general and presented
in an abstract manner. This abstract presentation hinders the direct
applicability of the theory. The purpose of this Appendix is to
provide a guide to the use of this theorem and to verify that the main
evolution system in this article satisfies the hypothesis of the
result. 

\medskip
Kato's theory is concerned with symmetric hyperbolic systems in which
the unknown $\mathbf{u}$ is regarded as a $\mathcal{P}$-valued
function over $\mathbb{R}^m$ where $\mathcal{P}$ is a Hilbert
space. The Hilbert space can be real or complex and, in fact, infinite
dimensional. In the present article we are interested in the case
where $\mathcal{P}$ is finite dimensional ---say, of dimension
$N$. In this case the symmetric hyperbolic system becomes a
\emph{standard} partial differential equation. For concreteness we set
here $\mathcal{P}=\mathbb{R}^N$ and $m=3$. The following discussion of Kato's
theorem will be made with this particular choice. in mind.

\medskip
Kato's theorem is concerned with ($N$-dimensional) symmetric hyperbolic quasilinear
systems of the form
\begin{equation}
\mathbf{A}^0(t,\underline{x},\mathbf{u}) \partial_t \mathbf{u} + \mathbf{A}^\alpha (t,\underline{x},\mathbf{u}) \partial_\alpha\mathbf{u} = \mathbf{F}(t,\underline{x},\mathbf{u}). 
\label{SHSQuasilinearSystem}
\end{equation}
for $0\leq t\leq T$, $\underline{x}\in\mathbb{R}^3$, $\alpha=1,\, 2,\,3$, and initial conditions
\begin{equation}
\mathbf{u}(0,x) = \mathbf{u}_\star(\underline{x}).
\label{QLSHSIC}
\end{equation}
In Kato's theory it is convenient to regard the coefficients
$\mathbf{A}^0(t,\underline{x},\mathbf{u})$ and $\mathbf{A}^\alpha
(t,\underline{x},\mathbf{u})$ as non-linear
operators depending on $t$ sending  $\mathbb{R}^N$-valued functions (i.e. the vector $\mathbf{u}$) over $\mathbb{R}^3$
into $(N\times N)$-matrix valued functions on
$\mathbb{R}^3$ ---in Kato's terminology these are the elements of
$\mathcal{B}(\mathcal{P})$, the space of bounded linear operators over
$\mathcal{P}$. Similarly, $\mathbf{F}(t,\underline{x},\mathbf{u})$ is
regarded as an non-linear operator depending on $t$ sending
$\mathbb{R}^N$-valued functions on $\mathbb{R}^3$ into
$\mathbb{R}^N$-valued functions on $\mathbb{R}^3$.

\medskip
Consider now $H^s(\mathbb{R}^3,\mathbb{R}^N)$, the space of
$(\mathbb{R}^N)$-vector valued functions over $\mathbb{R}^3$ such that
their entries have finite Sobolev norm of order $s$.  Let $\mathcal{D}$ be a bounded open subset of
$H^s(\mathbb{R}^3,\mathbb{R}^N)$. Writing
\[
\mathbf{A}^\mu(t,\underline{x},\mathbf{u}) =\big(
a^\mu_{ij}(t,\underline{x},\mathbf{u})\big), \qquad
\mathbf{F}(t,\underline{x},\mathbf{u}) = \big( f_i
(t,\underline{x},\mathbf{u})\big), \qquad i,\, j =1,\ldots N, \qquad \mu=0,\dots,3,
\]
one has that for fixed $t$ and $\mathbf{u}\in \mathcal{D}$ 
\begin{eqnarray*}
&& a^\mu_{ij}(t,\underline{x},\mathbf{u}): \mathbb{R}^m
   \rightarrow \mathbb{R}, \\
&& f_i (t,\underline{x},\mathbf{u}): \mathbb{R}^m
   \rightarrow \mathbb{R}.
\end{eqnarray*}

Key in Kato's analysis are the \emph{uniformly local Sobolev spaces}
$H^s_{ul}$. Let $C^\infty_0(\mathbb{R}^3,\mathbb{R})$ denote the sets
of smooth functions of compact support from $\mathbb{R}^3$ to
$\mathbb{R}$. Given any non-zero $\phi\in
C^\infty_0(\mathbb{R}^3,\mathbb{R})$ not identically zero, then
$\mathbf{u}\in H^s_{ul}$  if and only if 
\[
\sup_{\underline{x}\in \mathbb{R}^3} \parallel \phi_x
\mathbf{u}\parallel_s <\infty, \qquad \phi_x(\underline{y})\equiv \phi(\underline{y}-\underline{x}). 
\]

\begin{remark}
{\em In other words, the vector-valued function $\mathbf{u}$ is in
$H^s_{ul}$ if its Sobolev norm of order $s$ over any compact set over
$\mathbb{R}^3$ is finite and remains finite as one considers larger
and larger compact sets on $\mathbb{R}^3$.}
\end{remark}

\begin{remark}
{\em The spaces $H^s_{ul}$ satisfy nice embedding properties analogous
to those of $H^s$ ---see Lemma 2.7 in \cite{Kat75}.}
\end{remark}

In the following it will be assumed that, for fixed $t$ and $\mathbf{u}\in\mathcal{D}$, the
coefficients $a^\mu_{ij}(t,\underline{x},\mathbf{u}(\underline{x}))$
are functions from $\mathcal{D}$ to
$H^s_{ul}(\mathbb{R}^3,\mathbb{R})$. For $f_i
(t,\underline{x},\mathbf{u}(\underline{x}))$ one has the more relaxed
condition of being a function from $\mathcal{D}$ to
$H^s(\mathbb{R}^3,\mathbb{R})$. In Kato's more abstract terminology
this is equivalent to requiring that  $\mathbf{A}^\mu$ is a function
from $\mathcal{D}$ to
$H^s_{ul}(\mathbb{R}^3,\mathcal{B}(\mathcal{P}))$ and $\mathbf{F}$
from $\mathcal{D}$ to $H^s(\mathbb{R}^3,\mathcal{P})$. 

\medskip
One has the following reformulation of Theorem II in \cite{Kat75}:

\begin{theorem}
\label{Theorem:KatoII}
Let $s$ be a positive integer such that $s>3/2 +1=5/2$. Let
$\mathbf{A}^\mu(t,\underline{x},\mathbf{v}(\underline{x}))$,
$\mathbf{F}(t,\underline{x},\mathbf{v}(\underline{x}))$ and $\mathbf{v}\in\mathcal{D}$ as above
with $0\leq t\leq T$. Assume that the following conditions
hold:
\begin{itemize}
\item[(i)] The components
  $a^\mu_{ij}(t,\underline{x},\mathbf{v}(\underline{x}))$
  (respectively, $f_i
(t,\underline{x},\mathbf{v}(\underline{x}))$) are bounded in the
$H^s_{ul}$-norm (respectively $H^s$-norm) for $\mathbf{v}\in
\mathcal{D}$, uniformly in $t$.
 
\item[(ii)] For each $t$, the map $\mathbf{v}(\underline{x})\mapsto
  \mathbf{A}^\alpha(t,\underline{x},\mathbf{v}(\underline{x}))$ is
  uniformly Lipschitz continuous on $\mathcal{D}$ from the $H^0$-norm
  to the $H^0_{ul}$-norm, uniformly in $t$. Similarly, the map
  $\mathbf{v}(\underline{x})\mapsto
  \mathbf{F}(t,\underline{x},\mathbf{v}(\underline{x}))$ is Lipschitz
  continuous from the $H^0$-norm to the $H^0$-norm, again uniformly in
  $t$. 

\item[(iii)] The map $\mathbf{v}(\underline{x})\mapsto
  \mathbf{A}^0(t,\underline{x},\mathbf{v}(\underline{x}))$ is
  Lipschitz continuous on $\mathcal{D}$ from the $H^{s-1}$-norm to the
  $H^{s-1}_{ul}$-norm, uniformly in $t$. 

\item[(iv)] The maps
  $t\mapsto
  \mathbf{A}^\alpha(t,\underline{x},\mathbf{v}(\underline{x}))$ are
  continuous in the $H^0_{ul}$-norm for each
  $\mathbf{v}\in\mathcal{D}$. Similarly, the map $t\mapsto
  \mathbf{F}(t,\underline{x},\mathbf{v}(\underline{x}))$ is continuous
  in the $H^0$-norm for each $\mathbf{v}\in\mathcal{D}$. 

\item[(v)] The map $t\mapsto
  \mathbf{A}^0(t,\underline{x},\mathbf{v}(\underline{x}))$ is
  Lipschitz-continuous on $[0,T]$ in the $H^{s-1}_{ul}$-norm,
  uniformly for $\mathbf{v}\in\mathcal{D}$. 

\item[(vi)] For each $\mathbf{v}\in \mathcal{D}$ the matrix-valued
    functions
    $\mathbf{A}^\mu(t,\underline{x},\mathbf{v}(\underline{x}))$ are
    symmetric for each $(t,\underline{x})\in [0,T]\times
    \mathbb{R}^m$. 

\item[(vii)] The matrix
  $\mathbf{A}^0(t,\underline{x},\mathbf{v}(\underline{x}))$ is
  positive definite with eigenvalues larger that, say, $1$ for each
  $(t,\underline{x})$ and each $\mathbf{v}\in\mathcal{D}$. 

\item[(viii)] $\mathbf{u}_\star\in \mathcal{D}$.
\end{itemize}
Then there is a unique solution $\mathbf{u}$ to
\eqref{SHSQuasilinearSystem}-\eqref{QLSHSIC} defined on $[0,T']$ where
$0<T'\leq T$ such that
\[
\mathbf{u} \in C[0,T';\mathcal{D}]\cup C^1[0,T'; H^{s-1}(\mathbb{R}^3,\mathbb{R}^N)],
\]
where $T'$ can be chosen common to all initial conditions
$\mathbf{u}_\star$ in a suitably small condition of a given point in
$\mathcal{D}$. 
\end{theorem}

In practice, the conditions of the above theorem are hard to
verify. Kato provides sufficient conditions ensuring that conditions
in the above theorem are satisfied (Theorem IV in \cite{Kat75}:

\begin{theorem}
\label{Theorem:KatoSufficient}
Suppose that $s>3/2+1=5/2$. Let $\Omega$ be the subset of
$\mathbb{R}^3\times \mathbb{R}^N$ consisting of pairs
$(\underline{x},\underline{v})$ such that 
\[
| v -v_\star(x)|<\omega, \qquad \underline{x}\in \mathbb{R}^3
 \]
 where $\omega>0$ and 
$v_\star\in H^s(\mathbb{R}^3,\mathbb{R}^N)\subset C^1(\mathbb{R}^3,\mathbb{R}^N)$
are fixed. Let, as before, 
\begin{eqnarray*}
&& \mathbf{A}^\mu: [0,T]\times\Omega \longrightarrow \mathcal{B}(\mathbb{R}^N), \\
&& \mathbf{F}: [0,T]\times\Omega \longrightarrow \mathbb{R}^N,
\end{eqnarray*}
where $\mathcal{B}(\mathbb{R}^N)$ denotes the set of ($N\times
N$)-matrix valued functions over $\mathbb{R}^3$ with the properties
\begin{itemize}
\item[(a)] $\mathbf{A}^\alpha \in C[0,T;
  C^s_b(\Omega,\mathcal{B}(\mathbb{R}^N))]$, 
 \item[(b)] $\mathbf{A}^0\in   \mbox{\em Lip}[0,T;C^{s-1}_b(\Omega,\mathcal{B}(\mathbb{R}^N))]$,
\item[(c)] $\mathbf{F}\in C[0,T; C^{s+1}_b(\Omega,\mathbb{R}^N)]$,
\item[(d)] $\mathbf{F}_\star\in L^\infty[0,T; H^s(\mathbb{R}^3,\mathbb{R}^N)]\cap C[0,T; H^0(\mathbb{R}^3,\mathbb{R}^N)]$, 
\end{itemize}  
where $\mathbf{F}_\star(t,\underline{x})\equiv
\mathbf{F}(t,\underline{x},v_\star(\underline{x}))$. Then conditions
\emph{(i)}-\emph{(v)} in Theorem \ref{Theorem:KatoII} are satisfied
by $\mathbf{A}^\mu$, $\mathbf{F}$ provided that $\mathcal{D}$ is
chosen as a ball in $H^s(\mathbb{R}^3,\mathbb{R}^N)$ with $v_\star$ as
center and a sufficiently small radius $R_\star$. In addition,
\emph{(ix)} is satisfied if \emph{(a)} is assumed to hold with $s$
replaced by $s+1$. 
 \end{theorem}

\begin{remark}
{\em The sets $C^r_b(\Omega,\mathcal{B}(\mathbb{R}^N))$ and
  $C^r_b(\Omega,\mathbb{R}^N)$ denote the spaces of functions having
  derivatives up to the $r$-th order which are continuous and bounded
  in the supremum norm. 
}
\end{remark}

\begin{remark}
{\em If the $\mathbf{A}^\mu$ are polynomials in $\underline{p}$ it
actually suffices that the coefficients only be in $C[0,T;H^s_{ul}]$
and also in $C^1[0,T; H^{s-1}_{ul}]$ for $\mathbf{A}^0$. } 
\end{remark}

\section{Geodesic completeness of the background solution}
\label{Appendix:GeodesicCompleteness}

The geodesic completeness of the metric
\eqref{BackgroundPhysicalMetric} can be shown using the theory
developed in \cite{ChoCot02} ---in particular, Corollary 3.3 in this
reference applies to the present situation. 

\medskip
More precisely, the theory in \cite{ChoCot02} applies to spacetimes
$(\mathcal{M},\bmg)$ such that $\mathcal{M}=[t_\bullet,\infty)\times
\mathcal{S}$ where $t_\bullet>0$ and $\mathcal{S}$ is a smooth
 3-dimensional manifold. The metric $\bmg$ has
the $3+1$ split
\[
\bmg = -\alpha^2 \bmomega^0\otimes\bmomega^0 + h_{\bmi\bmj} \bmomega^\bmi\otimes\bmomega^\bmj,
\]
with 
\[
\bmomega^0 =\mathbf{d}t, \qquad \bmomega^\bmi = \mathbf{d}x^\bmi
+\beta^\bmi \mathbf{d}t.
\]
There exist numbers $0<\alpha_-,\, \alpha_+$ such that
\[
0< \alpha_- \leq \alpha \leq \alpha_+. 
\]
The metric $\bmh \equiv h_{\bmi\bmj}\mathbf{d}x^\bmi \otimes
\mathbf{d}x^\bmj$ is a \emph{geodesically complete Riemannian metric}
on $\mathcal{S}_t \equiv \{t\} \times \mathcal{S}$ such that there
exists a constant $C_1>0$ such that
\[
C_1 h_{\bmi\bmj}(t_\bullet) v^\bmi v^\bmj \leq h_{\bmi\bmj}(t) v^\bmi v^\bmj
\]
for all vectors on $T\mathcal{S}$ and $t\in[t_\bullet,\infty)$. Furthermore, there exists another
constant $C_2$ such that 
\[
\beta_i \beta^i \leq C_2, \qquad t\in[t_\bullet,\infty). 
\]
In the following let $K_{ij}$ denote the extrinsic curvature of the
hypersurfaces $\mathcal{S}_t$, $K_{\{ij\}}$ is tracefree part and $K$
its trace. 

\medskip
With the above conditions, the metric $\bmg$ is future geodesically
complete if the following two conditions hold:
\begin{itemize}
\item[(i)] $D_i \alpha D^i\alpha $ is bounded by a function of
  $t$ which is integrable on $[t_\bullet,\infty)$;
\item[(ii)]  $K<0$ and $K_{ij}K^{ij}$ is integrable on
  $[t_\bullet,\infty)$. 
\end{itemize}

\medskip
The metric \eqref{BackgroundPhysicalMetric} can be readily seen to
satisfy the above conditions. In particular, as $\alpha=1$, the norm
of the spatial gradient of the lapse vanishes and, accordingly, it is
integrable ---this verifies condition (i) above. Moreover, the extrinsic curvature of the hypersurfaces of
constant $t$ is given by
\[
K_{ij} = -\sinh t \cosh t \mathring{\gamma}_{ij},
\]
so that it is pure trace. Moreover, one has that 
\[
K = -3 \coth t <0, \qquad t\in[t_\bullet,\infty).
\]
As $K_{\{ij\}} =0$ in this case one has that (ii) is also
satisfied. It follows then that the background metric
$\mathring{\tilde{\bmg}}$ is future geodesically complete.

%%%%%%%%%%%%%%%%%

% Bibtex database
%\bibliography{/Users/Juan/Documents/tex/Newgrbib}

% Bibliography style
%\bibliographystyle{/Users/Juan/Documents/tex/reporthack}

\end{document}